\documentclass[manuscript,nonacm]{acmart}
\usepackage[utf8]{inputenc}
\usepackage{amsmath}
\usepackage{amsthm}
\usepackage{enumitem}
\usepackage{footnote}
\usepackage{algorithm,tabularx}
\usepackage{algpseudocode}
\newlength\myindent
\usepackage{xcolor}
\usepackage{tikz}
\usetikzlibrary{arrows.meta}
\usetikzlibrary{arrows,automata,calc,hobby}
\usetikzlibrary{shapes.geometric}
\usetikzlibrary{patterns}
\usetikzlibrary{intersections,through,backgrounds,matrix}
\usepackage{pgfplots}
\usepgfplotslibrary{fillbetween}
\pgfdeclarelayer{bg}
\pgfsetlayers{bg,main}

\tikzset{
    hatch distance/.store in=\hatchdistance,
    hatch distance=10pt,
    hatch thickness/.store in=\hatchthickness,
    hatch thickness=0.3pt
}
\makeatletter
\pgfdeclarepatternformonly[\hatchdistance,\hatchthickness]{northeast}
{\pgfqpoint{0pt}{0pt}}
{\pgfqpoint{\hatchdistance}{\hatchdistance}}
{\pgfpoint{\hatchdistance-1pt}{\hatchdistance-1pt}}%
{
    \pgfsetcolor{\tikz@pattern@color}
    \pgfsetlinewidth{\hatchthickness}
    \pgfpathmoveto{\pgfqpoint{0pt}{0pt}}
    \pgfpathlineto{\pgfqpoint{\hatchdistance}{\hatchdistance}}
    \pgfusepath{stroke}
}

\pgfdeclarepatternformonly[\hatchdistance,\hatchthickness]{northwest}
{\pgfqpoint{0pt}{0pt}}
{\pgfqpoint{\hatchdistance}{\hatchdistance}}
{\pgfpoint{\hatchdistance-1pt}{\hatchdistance-1pt}}%
{
    \pgfsetcolor{\tikz@pattern@color}
    \pgfsetlinewidth{\hatchthickness}
    \pgfpathmoveto{\pgfqpoint{\hatchdistance}{0pt}}
    \pgfpathlineto{\pgfqpoint{0pt}{\hatchdistance}}
    \pgfusepath{stroke}
}

\pgfdeclarepatternformonly[\hatchdistance,\hatchthickness]{horizontal}
{\pgfqpoint{0pt}{0pt}}
{\pgfqpoint{\hatchdistance}{\hatchdistance}}
{\pgfpoint{\hatchdistance-1pt}{\hatchdistance-1pt}}%
{
    \pgfsetcolor{\tikz@pattern@color}
    \pgfsetlinewidth{\hatchthickness}
    \pgfpathmoveto{\pgfqpoint{0pt}{0pt}}
    \pgfpathlineto{\pgfqpoint{\hatchdistance}{0pt}}
    \pgfusepath{stroke}
}

\pgfdeclarepatternformonly[\hatchdistance,\hatchthickness]{vertical}
{\pgfqpoint{0pt}{0pt}}
{\pgfqpoint{\hatchdistance}{\hatchdistance}}
{\pgfpoint{\hatchdistance-1pt}{\hatchdistance-1pt}}%
{
    \pgfsetcolor{\tikz@pattern@color}
    \pgfsetlinewidth{\hatchthickness}
    \pgfpathmoveto{\pgfqpoint{0pt}{0pt}}
    \pgfpathlineto{\pgfqpoint{0pt}{\hatchdistance}}
    \pgfusepath{stroke}
}

\pgfdeclarepatternformonly{mydots}{\pgfqpoint{-1pt}{-1pt}}{\pgfqpoint{1pt}{1pt}}{\pgfqpoint{3pt}{3pt}}
{%
  \pgfpathcircle{\pgfqpoint{0pt}{0pt}}{1pt}%
  \pgfusepath{fill}%
}%
\pgfdeclarepatternformonly[\LineSpace]{my north west lines}
  {\pgfqpoint{-0.2pt}{-0.2pt}}
  {\pgfpoint{\LineSpace+0.2pt}{\LineSpace+0.2pt}}
  {\pgfqpoint{\LineSpace}{\LineSpace}}%
{
  \pgfsetcolor{\tikz@pattern@color}
  \pgfsetlinewidth{0.4pt}
  \pgfpathmoveto{\pgfqpoint{0pt}{\LineSpace}}
  \pgfpathlineto{\pgfqpoint{\LineSpace}{0pt}}
  \pgfusepath{stroke}
}
\newdimen\LineSpace
\tikzset{
  line space/.code={\LineSpace=#1},
  line space=3pt
}
\makeatother

\algrenewcommand\algorithmicindent{1.0em}%
\usepackage{booktabs}
\usepackage{subcaption}
\newsavebox{\tempfig}
\newtheorem{theorem}{Theorem}
\newtheorem{lemma}[theorem]{Lemma}

\mathchardef\breakingcomma\mathcode`\,
{\catcode`,=\active
  \gdef,{\breakingcomma\discretionary{}{}{}}
}
\newcommand{\mathlist}[1]{\mathcode`\,=\string"8000 #1}

\long\def\com#1{}

\usepackage{xspace}
\newcommand{\prot}{\textsc{FnF-BFT}\xspace} 

\newcommand{\PACK}{\texttt{PACK}\xspace}
\newcommand{\ECO}{\texttt{EPOCHCHANGEOBJ}\xspace}

\newcommand{\PREPREPARE}{\texttt{PREPREPARE}\xspace}
\newcommand{\PREPARE}{\texttt{PREPARE}\xspace}
\newcommand{\PREPARED}{\texttt{PREPARED}\xspace}
\newcommand{\COMMIT}{\texttt{COMMIT}\xspace}
\newcommand{\COMMITTED}{\texttt{COMMITTED}\xspace}
\newcommand{\REQUEST}{\texttt{REQUEST}\xspace}
\newcommand{\ANSWER}{\texttt{ANSWER}\xspace}
\newcommand{\NEWCHECKPOINT}{\texttt{NEWCHECKPOINT}\xspace}
\newcommand{\CHECKPOINTED}{\texttt{CHECKPOINTED}\xspace}
\newcommand{\EPOCHCHANGE}{\texttt{EPOCHCHANGE}\xspace}
\newcommand{\NEWEPOCH}{\texttt{NEWEPOCH}\xspace}
\newcommand{\NEWEPOCHCONF}{\texttt{NEWEPOCHCONF}\xspace}
\tikzset{
short/.style={draw,rectangle,text height=3pt,text depth=13pt,
  text width=7pt,align=center,fill=gray!30},
long/.style={short,text width=1.5cm}
}

\def\lnode#1#2#3{%
  \node[long,right=of #1,label=center:#3] (#2) {}}

\usepackage[]{todonotes}

\begin{document}

\title{\prot: Exploring Performance Limits of BFT Protocols}

\author{Zeta Avarikioti}
\email{zetavar@ethz.ch}
\affiliation{\institution{ETH Z{\"u}rich}}
\author{Lioba Heimbach}
\email{hlioba@ethz.ch}
\affiliation{\institution{ETH Z{\"u}rich}}
\author{Roland Schmid}
\email{roschmi@ethz.ch}
\affiliation{\institution{ETH Z{\"u}rich}}
\author{Laurent Vanbever}
\email{lvanbever@ethz.ch}
\affiliation{\institution{ETH Z{\"u}rich}}
\author{Roger Wattenhofer}
\email{wattenhofer@ethz.ch}
\affiliation{\institution{ETH Z{\"u}rich}}
\author{Patrick Wintermeyer}
\email{patricwi@ethz.ch}
\affiliation{\institution{ETH Z{\"u}rich}}

\renewcommand{\shortauthors}{Avarikioti et al.}

\begin{abstract}
  We introduce \prot, a parallel-leader byzantine fault-tolerant state-machine replication protocol for the partially synchronous model with theoretical performance bounds during synchrony. By allowing all replicas to act as leaders and propose requests independently, \prot parallelizes the execution of requests. Leader parallelization distributes the load over the entire network -- increasing throughput by overcoming the single-leader bottleneck. We further use historical data to ensure that well-performing replicas are in command. \prot's communication complexity is linear in the number of replicas during synchrony and thus competitive with state-of-the-art protocols. Finally, with \prot, we introduce the first BFT protocol with performance guarantees in stable network conditions under truly byzantine attacks.
  A prototype implementation of \prot outperforms (state-of-the-art) HotStuff's throughput, especially as replicas increase, showcasing \prot's significantly improved scaling capabilities.
\end{abstract}





\maketitle
\section{Introduction}

\subsection{Motivation}

    In \textit{state machine replication (SMR)} protocols, distributed replicas aim to agree on a sequence of client requests in the presence of faults. To that end, SMR protocols rely strongly on another primitive of distributed computing, consensus.
    For protocols to maintain security under attack from malicious actors, consensus must be reached even when the replicas are allowed to send arbitrary information, namely under \textit{byzantine failures}. The protocols that offer these guarantees, i.e., are resilient against byzantine failures while continuing system operation, are known as \textit{byzantine  fault-tolerant (BFT)} protocols.

    The first practical BFT system, PBFT~\cite{castro2002practical}, was introduced more than two decades ago and has since sparked the emergence of numerous BFT systems~\cite{kotla2007zyzzyva,gueta2018sbft,yin2019hotstuff}. However, even today, BFT protocols do not scale well with the number of replicas, making large-scale deployment of BFT systems a challenge. Often, the origin of this issue stems from the \textit{single-leader bottleneck}: most BFT protocols rest the responsibility of executing client requests on a single leader instead of distributing it amongst replicas~\cite{stathakopoulou2019mir}. In such systems, the sole leader's hardware easily becomes overburdened with its duty as the central communication point of  message flow.

    Recently, protocols tackling the single-leader bottleneck through \textit{parallelization} emerged demonstrating staggering performance increases over state-of-the-art sequential-leader protocols~\cite{mao2008mencius,stathakopoulou2019mir,gupta2019scaling}. In the same fashion as most of their single leader counterparts, these works only consider non-malicious faults for the performance analysis. However, malicious attacks may lead to significant performance losses that are not evaluated. While these systems exhibit promising system performance with simple faults, they fail to lower-bound their performance in the face of malicious attacks from byzantine replicas.
    
    In this work, we propose \textsc{Fast'n'Fair-BFT} (\prot), a parallel-leader BFT protocol. 
    \prot circumvents the common single-leader bottleneck by utilizing parallel leaders to distribute the weight amongst all system replicas -- achieving a significant performance increase over sequential-leader systems. \prot scales well with the number of replicas and preserves \textit{high throughput even under arbitrarily malicious attacks from the byzantine replicas}.
    
    To establish this ability of our protocol, we define a new performance property, namely \textit{byzantine-resilient performance}, which encapsulates the ratio between the best-case and worst-case throughput of a BFT protocol, i.e., the effective utilization. Specifically, we bound this ratio to be constant, meaning that the throughput of a protocol under byzantine faults is lower-bounded by a constant fraction of the best-case throughput where no faults are present. We show that \prot achieves byzantine-resilient performance with a ratio of $16/27$ while maintaining \textit{safety} and \textit{liveness}. The analysis of \prot is conducted in the \textit{partially synchronous communication model}, meaning that a known bound $\Delta$ on message delivery holds after some unknown \textit{global stabilization time (GST)}. We further evaluate our protocol's efficiency by analyzing the amortized authenticator complexity after GST, similarly to HotStuff~\cite{yin2019hotstuff}.
    
    Finally, we provide a prototype implementation of \prot to demonstrate its scalability. Our implementation is based on  state-of-the-art Hostuff protocol~\cite{yin2019hotstuff}.
    \prot outperforms Hostuff's throughput by a factor rapidly increasing with the number of replicas, indicating remarkable improvement on scalability, while  
    exhibiting faster average performance.

\subsection{Related Work}
    Lamport et al.~\cite{lamport1982byzantine} first discussed the problem of reaching consensus in the presence of byzantine failures. Following its introduction, byzantine fault tolerance was initially studied in the synchronous network setting~\cite{pease1980reaching,dolev1982polynomial,dolev1985bounds}.
    Concurrently, the impossibility of deterministically reaching consensus in the asynchronous setting with a single replica failure was shown by Fischer et al.~\cite{fischer1985impossibility}. Dwork et al.~\cite{dwork1988consensus} proposed the concept of partial synchrony and demonstrated the feasibility of reaching consensus in partially synchronous networks. While the presented protocol always ensures safety, liveness relies on synchronous network conditions. During synchrony, the communication complexity of DSL is $\mathcal{O}(n^4 )$ -- making it unsuitable for deployment. In contrast to these works, \prot guarantees safety and liveness in partial synchrony, while the communication complexity is only $\mathcal{O}(n)$.
    
    Reaching consensus is needed to execute requests for state machine replication. Reiter~\cite{reiter1994secure,reiter1995rampart} introduced Rampart, an early protocol tackling byzantine fault tolerance for state machine replication. Rampart excludes faulty replicas from the group and replaces them with new replicas to make progress. Thus, Rampart relies on failure detection, which cannot be accurate in an asynchronous system, as shown by Lynch~\cite{lynch1996distributed}. \prot does not rely on failure detection. 

    With PBFT, Castro and Liskov~\cite{castro2002practical} devised the first efficient protocol for state machine replication that tolerates byzantine failures. The leader-based protocol requires $\mathcal{O}(n^2)$ communication to reach consensus, as well as $\mathcal{O}(n^3)$ for leader replacement. While widely deployed, PBFT does not scale well when the number of replicas increases. The quadratic complexity faced by the leader represents PBFT's bottleneck~\cite{bessani2014state}. While the PBFT implementation introduced by Behl et al.~\cite{behl2015consensus,behl2017hybster} is optimized for multi-cores, the complexity faced at the leader still presents the bottleneck of the state-of-the-art implementation. In this work, we tackle this issue by introducing $n$ parallel leaders that share the weight, thus efficiently alleviating the single leader's bottleneck.
    
    Kotla et al.~\cite{kotla2007zyzzyva} were the first to achieve $\mathcal{O}(n)$ complexity with Zyzzyva, an optimistic linear path PBFT. The complexity of leader replacement in Zyzzyva remains $\mathcal{O}(n^3)$, and safety violations were later exposed~\cite{abraham2017revisiting}. SBFT, devised by Gueta et al.~\cite{gueta2018sbft}, is a recent leader-based protocol that achieves $\mathcal{O}(n)$ complexity and improves the complexity of exchanging leaders to $\mathcal{O}(n^2)$. While reducing the overall complexity, the single leader is the bottleneck for both Zyzzyva and SBFT.

    Developed by Yin et al.~\cite{yin2019hotstuff}, leader-based HotStuff matches the $\mathcal{O}(n)$ complexity of Zyzzyva and SBFT. HotStuff rotates the leader with every request and is the first to achieve $\mathcal{O}(n)$ for leader replacement. However, HotStuff offers little parallelization, and experiments have revealed high complexity in practice~\cite{stathakopoulou2019mir}. While HotStuff's pipeline design offers an improvement over PBFT, its primary downside lies in the sequential proposal of requests and results in a lack of parallelism. On the contrary, $n$ parallel leaders propose requests simultaneously in \prot.

    Mao et al.~\cite{mao2008mencius,milosevic2013bounded} were the first to point out the importance of multiple leaders for high-performance state machine replication with Mencius and BFT-Mencius. Mencius maps client requests to the closest leader, and in turn, requests can become censored. However, no de-duplication measures are in place to handle the re-submission of censored client requests. \prot addresses this problem by periodically rotating leaders over the client space. 
    
    Gupta et al.~\cite{gupta2019scaling} recently introduced MultiBFT. MultiBFT is a protocol-agnostic approach to parallelize and improve existing BFT protocols. While allowing multiple instances to each run an individual client request, the protocol requires instances to unify after each request -- creating a significant overhead. Additionally, MultiBFT relies on failure detection, which is only possible in synchronous networks~\cite{lynch1996distributed}. With \prot, we allow leaders to make progress independently of each other without relying on failure detection. 
    
    Similarly, Stathakopoulou et al.~\cite{stathakopoulou2019mir} further investigated multiple leader protocols with Mir. Mir significantly improves throughput in comparison to sequential-leader approaches. However, as Mir runs instances of PBFT on a set of leaders, it incurs $\mathcal{O}(n^2)$ complexity, as well as $\mathcal{O}(n^3)$ complexity to update the leader set. We further expect Mir's performance to drop significantly in the presence of fully byzantine replicas, despite its high-throughput in the presence of crash failures.
    Mir updates the leader set as soon as a single leader in the set stops making progress -- allowing byzantine leaders to repeatedly end epochs early. \prot, however, continues to make progress in the presence of unresponsive byzantine leaders. We also show that the byzantine-resilient throughput is a constant fraction of the best-case throughput. 

    Byzantine resilience was initially explored by Clement et al.~\cite{clement2009making} who introduced Aardvark. Aardvark is an adaptation of PBFT with frequent view-changes: a leader only stays in its position when displaying an increasing throughput level. This first approach, however, comes with significant performance cuts in networks without failures. Parallel leaders allow \prot to be byzantine-resilient without accepting significant performance losses in an ideal setting.
    
    Byzantine resilience has further been studied since the introduction of Aardvark. Prime, proposed by Amir et al.~\cite{amir2008byzantine,amir2010prime}, aims to maximize performance in malicious environments. Besides, adding delay constraints that further confine the partially synchronous network model, Prime restricts its evaluation to delay attacks, i.e., the leader adds as much delay as possible to the protocol. Similarly, Veronese et al.~\cite{veronese2009spin} only evaluated their proposed protocol, Spinning, in the presence of delay attacks -- not fully capturing possible byzantine attacks. Consequently, the maximum performance degradation Spinning and Prime can incur under byzantine faults is at least 78\%~\cite{aublin2013rbft}. We analyze \prot theoretically to capture the entire spectrum of possible byzantine attacks.
    
    Aublin et al.~\cite{aublin2013rbft} further explored the performance of BFT protocols in the presence of byzantine attacks with RBFT. RBFT runs $f$ backup instances on the same set of client requests as the master instance to discover whether the master instance is byzantine. Thus, RBFT incurs quadratic communication complexity for every request. In this work, we reduce the communication complexity to $\mathcal{O}(n)$ and further increase performance through parallelization -- allowing byzantine-resilience without the added burden of detecting byzantine leaders. 

\subsection{Our Contribution}
    To the best of our knowledge, we introduce the first multiple leader BFT protocol with performance guarantees in stable network conditions under truly byzantine attacks, which we term \prot. Specifically, \prot is the first BFT protocol that achieves all the following properties:
    \begin{itemize}
        \item \textbf{Optimistic Performance:} After GST, the best-case throughput is $\Omega(n)$ times higher than the throughput of sequential-leader protocols.
        \item \textbf{Byzantine-Resilient Performance:} After GST, the worst-case throughput of the system is at least a constant fraction of its best-case throughput.
        \item \textbf{Efficiency:} After GST, the amortized authenticator complexity of reaching consensus is $\Theta(n)$.
    \end{itemize}
    
    We achieve these properties by combining two key components. First, we enable all replicas to continuously act as leaders in parallel to share the load of clients' requests. Second, unlike other protocols, we do not replace leaders upon failure but configure each leader's load based on the leader's past performance. With this combination, we guarantee a \textit{fair}    distribution of request according to each replica's capacity, which in turn results in \textit{fast} processing of requests.
    
    We further evaluate \prot's performance with a prototype implementation demonstrating its significantly improved scalability as well as its fast performance and high transaction throughput in comparison with state-of-the-art protocol HotStuff~\cite{yin2019hotstuff}. 
    
    The rest of the paper is structured as follows. We first define the model and the protocol goals (Section~\ref{sec:model}). Then, we introduce the design of \prot (Section~\ref{sec:prot}). Later, we present a security and performance analysis of our protocol (Section~\ref{sec:analysis}). 
    We conclude with an evaluation of \prot's performance (Section~\ref{sec:eval}).
    
\section{The Model}\label{sec:model}
    \subsection{System model}
    The system consists of $n=3f+1$ authenticated replicas and a set of clients. We index replicas by $i \in [n]=\{1, 2, \dots, n\}$.
    Throughout a protocol execution, at most $f$ unique replicas in the system are \textit{byzantine}, that is, instead of following the protocol they are controlled by an adversary with full information on their internal state.
    All other replicas are assumed to be \textit{correct}, i.e., following the protocol.
    Byzantine replicas may exhibit arbitrary adversarial behavior, meaning they can also behave like correct replicas.
    The adversary cannot intercept the communication between two correct replicas.
    Any number of clients may be byzantine.
    
    \subsection{Communication Model}
     We assume a \textit{partially synchronous communication model},i.e., a known bound $\Delta$ on message transmission will hold between any two correct replicas after some unknown \textit{global stabilization time (GST)}. 
     We show that \prot is safe in \textit{asynchrony}, that is, when messages between correct replicas are assumed to arrive in arbitrary order after any finite delay. We evaluate all other properties of the system after GST thus assuming a synchronous network.

    \subsection{Cryptographic Primitives}
        We make the usual cryptographic assumptions: the adversary is computationally bounded, and cryptographically-secure communication channels, computationally secure hash functions, (threshold) signatures, and encryption schemes exist.
        Similar to other BFT algorithms~\cite{amir2008steward,yin2019hotstuff,gueta2018sbft}, \prot makes use of threshold signatures. In a $(l,n)$ threshold signature scheme, there is a single public key held by all replicas and clients.
        Additionally, each replica~$u$ holds a distinct private key allowing to generate a partial signature~$\sigma_u(m)$ of any message~$m$.
        Any set of $l$ distinct partial signatures for the same message, $\left\{\sigma_u(m)\mid u \in U,  |U| = k\right\}$ can be combined (by any replica) into a unique signature $\sigma(m)$.
        The combined signature can be verified using the public key. We assume that the scheme is \textit{robust}, i.e., any verifier can easily filter out invalid signatures from malicious participants. In this work, we use a threshold $l=2f+1$.
        
    \subsection{Authenticator Complexity}
    \label{sec:authenticator-complexity}
        Message complexity has long been considered the main throughput-limiting factor in BFT protocols~\cite{gueta2018sbft,yin2019hotstuff}.
        In practice, however, the throughput of a BFT protocol is limited by both its computational footprint (mainly caused by cryptographic operations), as well as its message complexity. Hence, to assess the performance and efficiency of \prot, we adopt a complexity measure called authenticator complexity~\cite{yin2019hotstuff}.
        
        An \textit{authenticator} is any (partial) signature. We define the \textit{authenticator complexity} of a protocol as the number of all computations or verifications  of any authenticator done by replicas during the protocol execution.
        Note that the authenticator complexity also captures the message complexity of a protocol if, like in \prot, each message can be assumed to contain at least one signature.
        Unlike \cite{yin2019hotstuff}, where only the number of received signatures is considered, our definition allows to capture the load handled by each individual replica more accurately. Note that authenticator complexities according to the two definitions only differ by a constant factor.
        We only analyze the authenticator complexity after GST, as it is impossible for a BFT protocol to ensure deterministic progress and safety at the same time in an asynchronous network~\cite{fischer1985impossibility}.
        
     \subsection{Protocol Overview}
        The \prot protocol implements a state machine (cf.\ Section~\ref{sec:goals}) that is replicated across all replicas in the system. Clients broadcast requests to the system. Given client requests, replicas decide on the order of request executions and deliver commit-certificates to the clients.
        
        Our protocol moves forward in \textit{epochs}. In an epoch, each replica~$u$ is responsible for ordering a set of up to $C_u$ client requests that are independent of all requests ordered by other replicas in the epoch. Every replica in the system simultaneously acts as both a leader and a backup to the other leaders. The number of assigned client requests~$C_u$ is based on $u$'s past performance as a leader. During the epoch-change, a designated replica acting as primary: (a)~ensures that all replicas have a consistent view of the past leader and primary performance, (b)~deduces non-overlapping sequence numbers for each leader, and (c)~assigns parts of the client space to leaders.
        
        An epoch-change occurs whenever requested by more than two-thirds of the replicas. When seeking an epoch-change, a replica immediately stops participating in the previous epoch. The primary in charge of the epoch-change is selected through periodic rotation based on performance history. Replicas request an epoch-change if: (a)~all replicas~$u$ have exhausted their $C_u$ requests, (b)~a local timeout is exceeded, or (c)~enough other replicas request an epoch-change. Hence, epochs have bounded-length.
        
    \subsection{Protocol Goals}
    \label{sec:goals}
        \prot achieves scalable and byzantine fault-tolerant \textit{state machine replication (SMR)}. At the core of SMR, a group of replicas decide on a growing log of client requests. Clients are provided with cryptographically secure certificates which prove the commitment of their request.
        Fundamentally, the protocol ensures:
        \begin{enumerate}
            \item \textbf{Safety:} If any two correct replicas commit a request with the same sequence number, they both commit the same request.
            \item \textbf{Liveness:} If a correct client broadcasts a request, then every correct replica eventually commits the request.
        \end{enumerate}
        Thus, \prot will eventually make progress, and valid client requests cannot be censored.
        Additionally, \prot guarantees low overhead in reaching consecutive consensus decisions.
        Unlike other protocols limiting the worst-case efficiency for a single request, we analyze the amortized authenticator complexity per request after GST. We find this to be the relevant throughput-limiting factor:
        \begin{enumerate}[resume]
            \item \textbf{Efficiency:} After GST, the amortized authenticator complexity of reaching consensus is $\Theta(n)$.
        \end{enumerate}
        
        Furthermore, \prot achieves competitive performance under both optimistic and pessimistic adversarial scenarios:
        \begin{enumerate}[resume]
            \item \textbf{Optimistic Performance:} After GST, the best-case throughput is $\Omega(n)$ times higher than the throughput of sequential-leader protocols. 
            \item \textbf{Byzantine-Resilient Performance:} After GST, the worst-case throughput of the system is at least a constant fraction of its best-case throughput.
        \end{enumerate}
        Hence, unlike many other BFT systems, \prot guarantees that byzantine replicas cannot arbitrarily slow down the system when the network is stable.

\section{\prot}\label{sec:prot}
    \prot executes client requests on a state machine replicated across a set of $n$ replicas.
    We advance \prot in a succession of epochs -- identified by monotonically increasing \textit{epoch numbers}. Replicas in the system act as leaders and backups concurrently. As a leader, a replica is responsible for ordering client requests within its jurisdiction. Each leader~$v$ is assigned a predetermined number of requests~$C_v$ to execute during an epoch. To deliver a client request, $v$ starts by picking the next available sequence number and shares the request with the backups. Leader~$v$ must collect $2f+1$ signatures from replicas in the leader prepare and commit phase (Algorithm~\ref{pseudo}) to commit the request. We employ threshold signatures for the signature collection -- allowing us to achieve linear authenticator complexity for reaching consensus on a request. Additionally, we use low and high watermarks for each leader to represent a range of request sequence numbers that each leader can propose concurrently to boost individual leaders' throughput.

    Each epoch has a unique primary responsible for the preceding epoch-change, i.e., moving the system into the epoch. The replica elected as primary changes with every epoch and its selection is based on the system's history. A replica calls for an epoch-change in any of the following cases: (a)~the replica has locally committed requests for all sequence numbers available in the epoch, (b)~the maximum epoch time expired, (c)~the replica has not seen sufficient progress, or (d)~the replica has observed at least $f+1$ epoch-change messages from other replicas. 
    
    \prot generalizes PBFT~\cite{castro2002practical} to the $n$ leader setting. Additionally, we avoid PBFT's expensive all-to-all communication during epoch operation, similarly to Linear-PBFT~\cite{gueta2018sbft}. 
    Throughout this section, we discuss the various components of the protocol in further detail.
    
    \subsection{Client}
        Each client has a unique identifier. A client~$c$ requests the execution of an operation~$r$ by sending a $\mathlist{\langle \text{request}, r , t, c \rangle}$ to all leaders. Here, timestamp~$t$ is a monotonically increasing sequence number used to order the requests from one client. By using watermarks, we allow clients to have more than one request in flight. Client watermarks, low and high, represent the range of timestamp sequence numbers which the client can propose concurrently. Thus, we require $t$ to be within the low and high watermarks of client~$c$. The client watermarks are advanced similarly to the leader watermarks (cf.\ Section~\ref{sec:checkpoint}).
        Upon executing operation~$r$, replica~$u$ responds to the client with $\mathlist{\langle \text{reply}, e , d, u \rangle}$, where $e$ is the epoch number and $d$ is the request digest (cf.\ Section~\ref{sec:epoch-operation})\footnote{Instead of committing client request independently, the protocol could easily be adapted to process client requests in batches -- a standard BFT protocol improvement~\cite{kotla2007zyzzyva,yin2019hotstuff,stathakopoulou2019mir}.}. The client waits for $f+1$ such responses from the replicas.
    \subsection{Sequence Number Distribution}
        
        We distribute sequence numbers to leaders for the succeeding epoch during the epoch-change. While we commit requests from each leader in order,  the requests from different leaders are committed independently of each other in our protocol. Doing so allows leaders to continue making progress in an epoch, even though other leaders might have stopped working. Otherwise, a natural attack for byzantine leaders is to stop working and force the system to an epoch-change. Such attacks are possible in other parallel-leader protocols such as Mir~\cite{stathakopoulou2019mir}.
        
        To allow leaders to commit requests independently of each other, we need to allocate sequence numbers to all leaders during the epoch-change. Thus, we must also determine the number of requests each leader is responsible for before the epoch. The number of requests for leader~$v$ in epoch~$e$ is denoted by $C_v(e)$. It can be computed deterministically by all replicas in the network, based on the known history of the system (cf.\ Section~\ref{sec:epoch-change}).
       
        When assigning sequence numbers, we first automatically yield to each leader~$v \in [n]$ the sequence numbers of the $O_v(e)$ existing hanging operations from previous epochs in the assigned bucket(s). The remaining $C_v(e)-O_v(e)$ sequence numbers for each leader are distributed to them one after each other according to their ordering from the set of available sequence numbers. Note that $O_v(e)$ cannot exceed $C_v(e)$.
        For each leader~$v$ the assigned sequence numbers are mapped to local sequence numbers $1_{v,e}, 2_{v,e}, \dots, C_v(e)_{v,e}$ in epoch $e$. These sequence numbers are later used to simplify checkpoint creation (cf.\ Section~\ref{sec:checkpoint}).
        
    \subsection{Hash Space Division}\label{sec:hash}
    
        The request hash space is partitioned into buckets to avoid duplication. Each of these buckets is assigned to a single leader in every epoch. We consider the client identifier to be the request input and hash the client identifier ($h_c=h(c)$) to map requests into buckets. The hash space partition ensures that no two conflicting requests will be assigned to different leaders\footnote{Note that in case the requests are transactions with multiple inputs, the hash space division is more challenging to circumvent double-spending attacks. In such cases, we can employ well-known techniques~\cite{zamani2018rapidchain,kokoris2017omniledger} with no performance overhead as long as the average number of transactions' inputs remains constant~\cite{avarikioti2019divide}.}.
            
        Thus, the requests served by different leaders are independent of each other. Additionally, the bucket assignment is rotated round-robin across epochs, preventing request censoring. The hash space is portioned into $m\cdot n$ non-intersecting buckets of equal size, where $m\in \mathbb{Z}^+$ is a configuration parameter. Each leader~$v$ is then assigned $m_v(e)$ buckets in epoch~$e$ according to their load $C_v(e)$ (cf.\ Section~\ref{sec:epoch-change}). Leaders can only include requests from their active buckets. 
        
        When assigning buckets to leaders, the protocol ensures that every leader is assigned at least one bucket, as well as distributing the buckets according to the load handled by the leaders. Precisely, the number of buckets leader~$v$ is assigned in epoch~$e$ is given by 
        $m_v(e)= \left \lfloor{\dfrac{C_v(e)}{\sum _{u\in [n]} C_u(e)} (m-1)\cdot n}\right \rfloor+ 1 + \tilde{m}_v(e),$
        where $\tilde{m}_v(e)\in \{0,1\}$ distributes the remaining buckets to the leaders -- ensuring $\sum _{u \in [n]}m_u(e)=m \cdot n$. The remaining buckets are allocated to leaders~$v$ with the biggest value: 
        $\left \lfloor{\dfrac{C_v(e)}{\sum _{u\in [n]} C_u(e)} (m-1)\cdot n}\right \rfloor+ 1 - \dfrac{C_v(e)}{\sum _{u\in [n]} C_u(e)}\cdot m\cdot n.$
        
        Note that the system will require a sufficiently long stability period for all correct leaders to be working at their capacity limit, i.e., $C_v(e)$ matching the performance of leader~$v$ in epoch~$e$. Once correct leaders function at their capacity, the number of buckets they serve matches their capacity.
        The hash buckets are distributed to the leaders through a deterministic rotation such that each leader repeatedly serves each bucket under $f+1$ unique primaries. This rotation prevents byzantine replicas from censoring specific hash buckets. 
        For the remaining paper, we assume that there are always client requests pending in each bucket. Since we aim to optimize throughput, we consider this assumption in-sync with our protocol goals.
        
    \subsection{Primary Rotation}
        While all replicas are tasked with being a leader at all times, only a single replica, the primary, initiates an epoch. \prot assigns primaries periodically, exploiting the performance of good primaries and being reactive to network changes.
        The primary rotation consists of two core building blocks. First, \prot repeatedly rotates through the $2f+1$ best primaries and thus exploits their performance. Second, the primary assignment ensures that \prot explores every primary at least once within a sliding window. The sliding window consists of $g\in \mathbb{Z}$ epochs, and we set $g\geq 3f+1$ to allow the exploration of all primaries throughout a sliding window. We depict a sample rotation in Figure~\ref{fig:phases}. 

        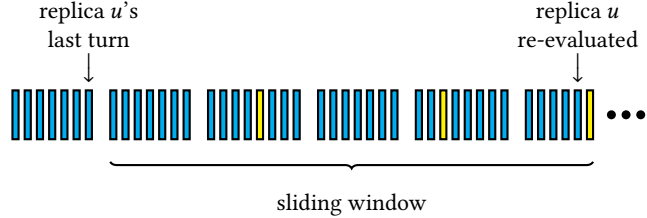
\begin{figure}[hbt]
            \centering
            \begin{tikzpicture}[scale=0.65,fill left half/.style={path picture={\fill[#1] (path picture bounding box.0) rectangle (path picture bounding box.150);}}, fill lower right/.style={path picture={\fill[#1] (path picture bounding box.south west) -- (path picture bounding box.north east) |-cycle;}}]
            \node[anchor =south,, align=center] at (7*0.25+0.0625,1.7) {replica~$u$'s\\last turn};
            \node[anchor =south] at (7*0.25+0.0625,1) {$\big\downarrow$};
            
            \node[anchor =south,, align=center] at (0.25+46*0.25+0.0625,1.7) {replica~$u$\\re-evaluated};
            \node[anchor =south] at (0.25+46*0.25+0.0625,1) {$\big\downarrow$};
            
            \foreach \x in {1,...,7} {
                \draw[thick,fill=cyan] (\x*0.25,0) rectangle (\x*0.25+0.125,1);
            }
            
            \foreach \x in {1,...,7} {
                \draw[thick,fill=cyan] (\x*0.25+8*0.25,0) rectangle (\x*0.25+8*0.25+0.125,1);
            }
            \foreach \x in {1,...,4} {
                \draw[thick,fill=cyan] (\x*0.25+16*0.25,0) rectangle (\x*0.25+16*0.25+0.125,1);
            }
            
            \foreach \x in {1,...,1} {
                 \draw[thick,fill=yellow] (\x*0.25+20*0.25,0) rectangle (\x*0.25+20*0.25+0.125,1);
            }
            
            \foreach \x in {1,...,3} {
                \draw[thick,fill=cyan] (\x*0.25+21*0.25,0) rectangle (\x*0.25+21*0.25+0.125,1);
            }
            
            \foreach \x in {1,...,7} {
                \draw[thick,fill=cyan] (\x*0.25+25*0.25,0) rectangle (\x*0.25+25*0.25+0.125,1);
            }
            
            \foreach \x in {1,...,2} {
                \draw[thick,fill=cyan] (\x*0.25+33*0.25,0) rectangle (\x*0.25+33*0.25+0.125,1);
            }
            \foreach \x in {1,...,1} {
                 \draw[thick,fill=yellow] (\x*0.25+35*0.25,0) rectangle (\x*0.25+35*0.25+0.125,1);
            }
            \foreach \x in {1,...,5} {
                \draw[thick,fill=cyan] (\x*0.25+36*0.25,0) rectangle (\x*0.25+36*0.25+0.125,1);
            }
            
            \foreach \x in {1,...,5} {
                \draw[thick,fill=cyan] (\x*0.25+42*0.25,0) rectangle (\x*0.25+42*0.25+0.125,1);
            }
            \foreach \x in {1,...,1} {
                 \draw[thick,fill=yellow] (\x*0.25+47*0.25,0) rectangle (\x*0.25+47*0.25+0.125,1);
            }

                \node at (12.5,0.5) [circle,fill,inner sep=1.2pt]{};
                \node at (12.8,0.5) [circle,fill,inner sep=1.2pt]{};
                \node at (13.1,0.5) [circle,fill,inner sep=1.2pt]{};

                \draw [
                    thick,
                    decoration={
                        brace,
                        mirror,
                        raise=0.3cm
                    },
                    decorate
                ] (0.25+8*0.25,0) -- (0.25+47*0.25+0.125,0) 
                node [pos=0.5,anchor=north,yshift=-0.6cm] {sliding window}; 
            \end{tikzpicture}
            
            \caption{\prot primary rotation in a system with $n=10$ replicas. In blue, we show epochs led by primaries elected based on their performance. Epochs shown in yellow are led by replicas re-evaluated once their last turn as primary falls out of the sliding window.}
            \label{fig:phases}
        \end{figure}

        Throughout the protocol, all replicas record the performance of each primary. We measure performance as the number of requests successfully committed under a primary in an epoch. Performance can thus be determined during the succeeding epoch-change by each replica (cf.\ Section~\ref{sec:epoch-change}). To deliver a reactive system, we update a replica's primary performance after each turn.
        
        We rotate through the best $2f+1$ primaries repeatedly. After every $2f+1$ primaries, the best $2f+1$ primaries are redetermined and subsequently elected as primary in order of the time passed since their last turn as primary. The primary that has not been seen for the longest time is elected first. 
        Cycling through the best primaries maximizes system performance. Simultaneously, basing performance solely on a replica's preceding primary performance strips byzantine primaries from the ability to misuse a good reputation.
        Every so often, we interrupt the continuous exploitation of the best $2f+1$ primaries to revisit replicas that fall out of the sliding window. If replica~$u$'s last turn as primary occurred in epoch $e-g$ by the time epoch $e$ rolls around, replica~$u$ would be re-explored as primary in epoch $e$. The exploration allows us to re-evaluate all replicas as primaries periodically and ensures that \prot is reactive to network changes.
        
        Note that we start the protocol by exploring all primaries ordered by their identifiers. We would also like to point out that only one primary can fall out of the sliding window at any time after the initial exploration. Thus, we always know which primary will be re-evaluated. 
        
    \subsection{Epoch Operation}\label{sec:epoch-operation}
        To execute requests, we use a leader-based adaption of PBFT, similar to Linear-PBFT~\cite{gueta2018sbft}. Threshold signatures are commonly used to reduce the complexity of the backup prepare and commit phases of PBFT. The leader of a request is used as a collector of partial signatures to create a $(2f+1,n)$ threshold signature in the intermediate stages of the backup prepare and commit phases. We visualize the schematic of the message flow for one request led by replica~$0$ in Figure~\ref{fig:messageflow} 
        and  summarize the protocol executed locally by replicas to deliver a request proposed by leader~$v$ in Algorithm~\ref{pseudo}. 
            
        \begin{figure}[hbt!]
            \centering
            \begin{tikzpicture}[scale = 0.7]

                \foreach \c in {2,4}{
                    \draw[dotted,thick] (2.2*\c,3)--++(0,-3);
                }
                \node[] at (1.1*3,3.5) {\small leader prepare};
                \node[] at (1.1*6,3.5) {\small backup prepare};
                \node[] at (1.1*10,3.5) {\small commit};
                 	 	
                \foreach \a/\b/\c/\d in {
                    1.2/3/0.8/-1,
                    1.2/3/0.7/-2,
                    1.2/3/0.6/-3,
                    2.1/2/0.6/1,
                    2.2/1/0.6/2,
                    2.0/0/0.9/3,
                    3.2/3/0.8/-1,
                    3.2/3/0.7/-2,
                    3.2/3/0.6/-3,
                    4.1/2/0.6/1,
                    4.1/1/0.7/2,
                    4.0/0/0.9/3,
                    5.2/3/0.8/-1,
                    5.2/3/0.7/-2,
                    5.2/3/0.7/-3
                }{ \draw[cyan,thick,-{Stealth}] (2.2*\a,\b)--++(2.2*\c,\d); }
                    
                \foreach \row [count=\r] in {\small 0,\small 1,\small 2,\small 3} {
                    \draw[very thick](2.2,4-\r)node[left]{\row} -- ++(1.1*10,0);
                }

            \end{tikzpicture}
            \caption{Schematic message flow for one request.}
            \label{fig:messageflow}
        \end{figure}
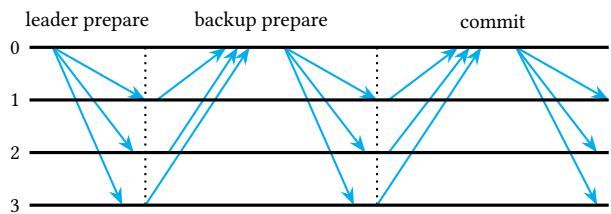
        
        \paragraph{Leader prepare phase.} Upon receiving a $\mathlist{\langle \text{request}, r , t, c \rangle}$ from a client, each replica computes the hash of the client identifier~$c$. If the request falls into one of the active buckets belonging to leader~$v$, $v$ verifies $\mathlist{\langle \text{request}, r , t, c \rangle}$ from client~$c$. The request is discarded, if (a)~it has already been prepared, or (b)~it is already pending. Once verified, leader~$v$ broadcasts $\mathlist{\langle \text{pre-prepare}, sn , e, h(r), v \rangle}$, where $sn$ is the sequence number, $e$ the current epoch, $h(r)$ is the hash digest of request~$r$ and $v$ represents the leader's signature. The cryptographic hash function~$h$ maps an arbitrary-length input to a fixed-length output. We can use the digest~$h(r)$ as a unique identifier for a request~$r$, as we assume the hash function to be collision-resistant.
            
        \paragraph{Backup prepare phase.} A backup~$w$ accepts $\mathlist{\langle \text{pre-prepare}, sn , e, h(r), v \rangle}$ from leader~$v$, if (a)~the epoch number matches its local epoch number, (b)~$w$ has not prepared another request with the same sequence number~$sn$ in epoch~$e$, (c)~leader~$v$ leads sequence number~$sn$, (d)~$sn$ lies between the low and high watermarks of leader~$v$, (e)~$r$ is in the active bucket of $v$, and (f)~$r$ was submitted by an authorized client. Upon accepting $\mathlist{\langle \text{pre-prepare}, sn , e, h(r), v \rangle}$, $w$ computes $d= h(sn \| e\| r)$ where $h$ is a cryptographic hash function. Additionally, $w$ signs $d$ by computing a verifiable partial signature~$\sigma_w(d)$. Then $w$ sends $\mathlist{\langle \text{prepare}, sn , e, \sigma_w(d) \rangle}$ to leader~$v$. Upon receiving $2f$ prepare messages for $sn$ in epoch~$e$, leader~$v$ forms a combined signature~$\sigma (d)$ from the $2f$ prepare messages and its own signature. Leader~$v$ then broadcasts $\mathlist{\langle \text{prepared-certificate}, sn , e, \sigma(d) \rangle}$ to all backups.

        \paragraph{Commit phase.} Backup~$w$ accepts the \textit{prepared-certificate} and replies $\mathlist{\langle \text{commit},\allowbreak sn ,\allowbreak e,\allowbreak \sigma _w( \sigma(d) )\rangle}$ to leader~$v$. After collecting $2f$ commit messages, leader~$v$ creates a combined signature~$\sigma (\sigma (d))$ using the signatures from the collected commit messages and its own signature. Once the combined signature is prepared, $v$ continues by broadcasting $\mathlist{\langle\text{commit-certificate}, sn , e, \sigma( \sigma(d) )\rangle}$. Upon receiving the \textit{commit-certificate}, replicas execute $r$ after delivering all preceding requests led by $v$, and send replies to the client.
        
\subsection{Checkpointing}\label{sec:checkpoint}
        
        Similar to  PBFT~\cite{castro2002practical}, we periodically create checkpoints to prove the correctness of the current state. Instead of requiring a costly round of all-to-all communication to create a checkpoint, we add an intermediate phase and let the respective leader collect partial signatures to generate a certificate optimistically. Additionally, we expand the PBFT checkpoint protocol to run for $n$ parallel leaders. 
        
        For each leader~$v$, we repeatedly create checkpoints to clear the logs and advance the watermarks of leader~$v$ whenever the local sequence number~$sn_{v,e,k}$ is divisible by a constant $k\in \mathbb{Z}^+$. Recall that when a replica~$u$ delivers a request for leader~$v$ with local sequence number~$sn_{v,e,k}$, this implies that all requests led by $v$ with local sequence number lower than $sn_{v,e,k}$ have been locally committed at replica~$u$. Hence, after delivering the request with local sequence number~$sn_{v,e,k}$, replica~$u$ sends $\mathlist{\langle \text{checkpoint}, sn_{v,e,k}, h(sn'_{v,e,k}), u\rangle}$ to leader~$v$. Here, $sn'_{v,e,k}$ is the last checkpoint and $h(sn'_{v,e,k})$ is the hash digest of the requests with sequence number~$sn_v$ in the range $sn'_{v,e,k}\leq sn_v\leq sn_{v,e,k}$. Leader~$v$ proceeds by collecting $2f+1$ checkpoint messages (including its own) and generates a \textit{checkpoint-certificate} by creating a combined threshold signature. 
        Then, leader $v$ sends the checkpoint-certificate to all other replicas. If a replica sees the checkpoint-certificate, the checkpoint is \textit{stable} and the replica can discard the corresponding messages from its logs, i.e., for sequence numbers belonging to leader~$v$ lower than $sn_{v,e,k}$.
        
        We use checkpointing to advance low and high watermarks. In doing so, we allow several requests from a leader to be in flight. The low watermark~$L_v$ for leader~$v$ is equal to the sequence number of the last stable checkpoint, and the high watermark is $H_v=L_v+2k$. We set $k$ to be large enough such that replicas do not stall. Given its watermarks, leader~$v$ can only propose requests with a local sequence number between low and high watermarks.

 
\begin{figure}[t]
    \vspace{-0.5cm}
    \begin{minipage}[t]{0.48\textwidth}    
        \begin{algorithm}[H]
            \centering
            \caption{Committing a request proposed by leader~$v$}\label{pseudo}
            \begin{algorithmic}
                \State \hspace{-\algorithmicindent}\textit{Leader prepare phase}
                \State \textbf{as} replica~$u$: 
                \State \hspace{\algorithmicindent}\textbf{upon} receiving a valid  $\mathlist{\langle \text{request}, r , t, c \rangle}$ from client~$c$:  
                \State \hspace{\algorithmicindent}\hspace{\algorithmicindent} map client request to hash bucket
                \State \textbf{as} leader~$v$: 
                \State \hspace{\algorithmicindent}accept $\mathlist{\langle \text{request}, r , t, c \rangle}$ assigned to one of $v$'s buckets 
                \State \hspace{\algorithmicindent}pick next assigned sequence number~$sn$
                \State \hspace{\algorithmicindent}broadcast $\mathlist{\langle \text{pre-prepare}, sn , e, h(r), v \rangle}$ 
                \State \hspace{-\algorithmicindent}\textit{Backup prepare phase}
                \State \textbf{as} backup~$w$:
                \State \hspace{\algorithmicindent}accept $\mathlist{\langle \text{pre-prepare}, sn , e, h(r), v \rangle}$ 
                \State \hspace{\algorithmicindent}\textbf{if} the pre-prepare message is valid: 
                \State \hspace{\algorithmicindent}\hspace{\algorithmicindent}compute partial signature~$\sigma_w(d)$
                \State \hspace{\algorithmicindent}\hspace{\algorithmicindent}send $\mathlist{\langle \text{prepare}, sn , e, \sigma_w(d) \rangle}$ to leader~$v$
                \State \textbf{as} leader~$v$:
                \State \hspace{\algorithmicindent}compute partial signature~$\sigma_v(d)$
                \State \hspace{\algorithmicindent}\textbf{upon} receiving $2f$ prepare messages: 
                \State \hspace{\algorithmicindent}\hspace{\algorithmicindent}compute $(2f+1,n)$ threshold signature~$\sigma(d)$
                \State \hspace{\algorithmicindent}\hspace{\algorithmicindent}broadcast $\mathlist{\langle \text{prepared-certificate}, sn , e, \sigma(d) \rangle}$ 
                \State \hspace{-\algorithmicindent}\textit{Commit phase}
                \State \textbf{as} backup~$w$:
                \State \hspace{\algorithmicindent}accept $\mathlist{\langle \text{prepared-certificate}, sn , e, \sigma(d) \rangle}$
                \State \hspace{\algorithmicindent}compute partial signature~$\sigma(\sigma_w(d))$
                \State \hspace{\algorithmicindent}$\mathlist{\langle \text{commit}, sn , e, \sigma _w( \sigma(d) )\rangle}$ to leader~$v$ 
                \State \textbf{as} leader~$v$:
                \State \hspace{\algorithmicindent}compute partial signature~$\sigma(\sigma_v(d))$
                \State \hspace{\algorithmicindent}\textbf{upon} receiving $2f$ commit messages: 
                \State \hspace{\algorithmicindent}\hspace{\algorithmicindent}compute $(2f+1,n)$ threshold signature~$\sigma(\sigma(d))$
                \State \hspace{\algorithmicindent}\hspace{\algorithmicindent}broadcast $\mathlist{\langle \text{commit-certificate}, sn , e, \sigma( \sigma(d) )\rangle}$ 
            \end{algorithmic}
        \end{algorithm}
    \end{minipage}
    \hfill
    \begin{minipage}[t]{0.48\textwidth} 
         \begin{algorithm}[H]
            \caption{Epoch-change protocol for epoch~$e+1$}\label{epoch-changeproto}
            \begin{algorithmic}
                \State \hspace{-\algorithmicindent}\textit{Starting epoch-change}
                \State \textbf{as} replica~$u$: 
                \State \hspace{\algorithmicindent}broadcast $\mathlist{\langle \text{epoch-change}, e+1, \mathcal{S}, \mathcal{C}, \mathcal{P}, \mathcal{Q},u \rangle}$
                \State \hspace{\algorithmicindent}\parbox[t]{0.9\textwidth}{\textbf{upon} receiving $2f$  epoch-change messages for $e+1$:}
                \State \hspace{\algorithmicindent}\hspace{\algorithmicindent}start epoch-change timer~$T_e$
                \State \hspace{-\algorithmicindent}\textit{Reliable broadcast}
                \State \textbf{as} primary $p_{e+1}$:
                \State \hspace{\algorithmicindent}compute $C_v(e+1)$ for all leaders $v\in [n]$
                \State \hspace{\algorithmicindent}\parbox[t]{0.9\textwidth}{perform 3-phase reliable broadcast sharing configuration details of epoch~$e+1$ and the performance of  primary $p_{e}$}
                \vspace{1mm}
                \State \textbf{as} replica~$u$:
                \State \hspace{\algorithmicindent}participate in reliable broadcast initiates by $p_{e+1}$
                \State \hspace{-\algorithmicindent}\textit{Starting epoch}
                \State \textbf{as} primary $p_{e+1}$:
                \State \hspace{\algorithmicindent}broadcast $\mathlist{\langle \text{new-epoch},e+1,\mathcal{V},\mathcal{O},\allowbreak p_{e+1}\rangle}$
                \State \hspace{\algorithmicindent}enter epoch~$e+1$
                \State \textbf{as} replica~$u$:
                \State \hspace{\algorithmicindent}accept $\mathlist{\langle \text{new-epoch},e+1,\mathcal{V},\mathcal{O},\allowbreak p_{e+1}\rangle}$
                \State \hspace{\algorithmicindent}enter epoch~$e+1$ 
            \end{algorithmic}
        \end{algorithm}
        \vspace{-2.1mm}
        \begin{algorithm}[H]
            \centering
            \caption{Configuration adjustment}\label{multipleleadersnotimeout}
            \begin{algorithmic}
                \State \hspace{-\algorithmicindent}initially $C_v(1)=C_{\min}$ for all replicas $v$ 
                \State \textbf{if} $c_v(e)< C_v(e)$
                \State $C_v(e+1) = \max\left(C_{\min},\max_{i\in \{0, \dots, f\}} \left(c_v(e-i)\right)\right)$
                \State \textbf{else}
                \State $C_v(e+1) = 2\cdot c_v(e)$
            \end{algorithmic}
        \end{algorithm}
    \end{minipage}
\end{figure}

    \subsection{Epoch-Change}\label{sec:epoch-change}
    
        At a high level, we modify the PBFT epoch-change protocol as follows: we use threshold signatures to reduce the message complexity and extend the epoch-change message to include information about all leaders. Similarly to Mir~\cite{stathakopoulou2019mir}, we introduce a round of reliable broadcast to share information needed to determine the configuration of the next epoch(s). In particular, we determine the load assigned to each leader in the next epoch, based on their past performance. We also record the performance of the preceding primary. An overview of the epoch-change protocol can be found in Algorithm~\ref{epoch-changeproto}, while a detailed description follows.

        \paragraph{Calling epoch-change.} Replicas call an epoch-change by broadcasting an epoch-change message in four  cases:
        \begin{enumerate}
            \item Replica~$u$ triggers an epoch-change in epoch~$e$, once it has committed everyone's assigned requests locally.
            \item Replica~$u$ calls for an epoch-change when its \textit{epoch timer} expires. The value of the epoch timer~$T$ is set to ensure that after GST, correct replicas can finish at least $C_{\min}$ requests during an epoch. $C_{\min}\in \Omega(n^2)$ is the minimum number of requests assigned to leaders. 
            \item Replicas call epoch-changes upon observing inadequate progress. Each replica~$u$ has individual \textit{no-progress timers} for all leaders. The no-progress timer is initialized with the same value $T_p$ for all leaders. Initially, replicas set all no-progress timers for the first time after $5\Delta$ in the epoch -- accounting for the message transmission time of the initial requests. A replica resets the timer for leader~$v$ every time it receives a commit-certificate from $v$. In case the replica has already committed $C_v$ requests for leader~$v$, the timer is no longer reset. Upon observing no progress timeouts for $b  \in [f+1,2f+1]$ different leaders, a replica calls an epoch-change. Requiring at least $f+1$ leaders to make progress ensures that a constant fraction of leaders makes progress, and at least one correct leader is involved. On the other hand, we demand no more than $2f+1$ leaders to make progress such that byzantine leaders failing to execute requests cannot stop the epoch early. We let $b=2f+1$ and set the no-progress timer such that it does not expire for correct leaders and simultaneously ensures sufficient progress, i.e., $T_p\in \Theta(T/C_{\min})$.
            \item Finally, replica~$u$ calls an epoch-change if it sees that $f+1$ other replicas have called an epoch-change for an epoch higher than $e$. Then, replica $u$ picks the smallest epoch in the set such that byzantine replicas cannot advance the protocol an arbitrary number of epochs. 
        \end{enumerate}
        After sending an epoch-change message, the replica will only start its epoch-change timer, once it saw at least $2f+1$ epoch-change messages. We will discuss the epoch-change timer in more detail later.
        
        \paragraph{Starting epoch-change (Algorithm~\ref{epoch-changeproto}, steps 1-5).}
        To move the system to epoch~$e+1$, replica~$u$ sends $\mathlist{\langle \text{epoch-change}, e+1, \mathcal{S}, \mathcal{C}, \mathcal{P}, \mathcal{Q},u \rangle}$ to all replicas in the system. Here, $\mathcal{S}$ is a vector of sequence numbers $sn_v$ of the last stable checkpoints~$S_v$ $\forall v \in [n]$ known to $u$ for each leader~$v$. $\mathcal{C}$ is a set of checkpoint-certificates proving the correctness of $S_v$ $\forall v \in [n]$, while $\mathcal{P}$ contains sets $\mathcal{P}_v$ $\forall v \in [n]$. For each leader~$v$, $\mathcal{P}_v$ contains a prepared-certificate for each request~$r$ that was prepared at $u$ with sequence number higher than $sn_v$, if replica~$v$ does not possess a commit-certificate for $r$. Similarly, $\mathcal{Q}$ contains sets $\mathcal{Q}_v$ $\forall v \in [n]$. $\mathcal{Q}_v$ consists of a commit-certificate for each request~$r$ that was prepared at $u$ with sequence number higher than $sn_v$.
        
        \paragraph{Reliable broadcast (Algorithm~\ref{epoch-changeproto}, steps 6-11).}    
        The primary of epoch~$e+1$ ($p_{e+1}$) waits for $2f$ epoch-change messages for epoch~$e$. Upon receiving a sufficient number of messages, the primary performs a classical 3-phase reliable broadcast. During the broadcast, the primary informs leaders on the number of requests assigned to each leader in the next epoch and the identifiers of the replicas which send epoch-change messages. The number of requests assigned to a leader is computed deterministically (Algorithm~\ref{multipleleadersnotimeout}). Through the reliable broadcast, we ensure that the primary cannot share conflicting information regarding the sequence number assignment and, in turn, the next epoch's sequence number distribution. In addition to sharing information about the epoch configuration, the primary also broadcasts the total number of requests committed during the previous epoch. This information is used by the network to evaluate primary performance and determine epoch primaries. 
        
        \paragraph{Starting epoch (Algorithm~\ref{epoch-changeproto}, steps 12-18).}
        The primary~$p_{e+1}$ multicasts $\mathlist{\langle \text{new-epoch},e+1,\mathcal{V},\mathcal{O},\allowbreak p_{e+1}\rangle}$. Here, the set~$\mathcal{V}$ contains sets $\mathcal{V}_u$, which carry the valid epoch-change messages of each replica~$u$ of epoch~$e$ received by the primary of epoch~$e+1$, plus the epoch-change message the primary of epoch~$e+1$ would have sent. $\mathcal{O}$ consists of sets $\mathcal{O}_v$ $\forall v \in [n]$ containing pre-prepare messages and commit-certificates.
        
        $\mathcal{O}_v$ is computed as follows. First, the primary determines the sequence number~$S_{\min}(v)$ of the latest stable checkpoint in $\mathcal{V}$ and the highest sequence number~$S_{\max}(v)$ in a prepare message in $\mathcal{V}$. For each sequence number~$sn_v$ between $S_{\min}(v)$ and $S_{\max}(v)$ of all leaders $v\in[n]$ there are three cases: (a)~there is at least one set in $\mathcal{Q}_v$ of some epoch-change message in $\mathcal{V}$ with sequence number~$sn_v$, (b)~there is at least one set in $\mathcal{P}_v$ of some epoch-change message in $\mathcal{V}$ with sequence number~$sn_v$ and none in  $\mathcal{Q}_v$, or (c)~there is no such set. In the first case, the primary simply prepares a commit-certificate it received for $sn_v$. In the second case, the primary creates a new message $\mathlist{\langle \text{pre-prepare},  sn_v, e+1, d,p_{e+1}\rangle}$, where $d$ is the request digest in the pre-prepare message for sequence number~$sn_v$ with the highest epoch number in $\mathcal{V}$. In the third case, the primary creates a new pre-prepare message $\mathlist{\langle \text{pre-prepare}, sn_v,e+1,  d^{null},p_{e+1}\rangle}$, where $d^{null}$ is the digest of a special null request; a  null request goes through the protocol like  other requests, but its execution is a no-op. If there is a gap between $S_{\max}(v)$ and the last sequence number assigned to leader~$v$ in epoch~$e$, these sequence numbers will be newly assigned in the next epoch.
            
        Next, the primary appends the messages in $\mathcal{O}$ to its log. If $S_{\min}(v)$ is greater than the sequence number of its latest stable checkpoint, the primary also inserts the proof of stability (the checkpoint with sequence number~$S_{\min}(v)$) in its log. Then it enters epoch~$e+1$; at this point, it can accept messages for epoch~$e+1$.
            
        A replica accepts a new-epoch message for epoch~$e+1$ if: (a)~it is signed properly, (b)~the epoch-change messages it contains are valid for epoch~$e+1$,~(c) the information in $\mathcal{V}$ matches the new request assignment, and (d)~the set~$\mathcal{O}$ is correct. The replica verifies the correctness of $\mathcal{O}$ by performing a computation similar to the one previously used by the primary. Then, the replica adds the new information contained in $\mathcal{O}$ to its log and decides all requests for which a commit-certificate was sent. Replicas rerun the protocol for messages between $S_{\min}(v)$ and $S_{\max}(v)$ without a commit-certificate. They do not execute client requests again (they use their stored information about the last reply sent to each client instead). As request messages and stable checkpoints are not included in new-epoch messages, a replica might not have some of them available. In this case, the replica can easily obtain the missing information from other replicas in the system.
        
        \paragraph{Hanging requests.}
        While the primary sends out the pre-prepare message for all hanging requests, replicas in whose buckets the requests fall, are responsible for computing prepared- and commit-certificates of the individual requests. 
        In the example shown in Figure~\ref{fig:messageflowepochchange}, the primary of epoch~$e+1$, replica~$0$, sends a pre-prepare message for a request in a bucket of replica~$1$, contained in the new-epoch message, to everyone. Replica~$1$ is then responsible for prepared- and commit-certificates, as well as collecting the corresponding partial signatures.

        The number of request~$C_v(e+1)$ assigned to leader~$v$ in epoch~$e+1$ is determined deterministically based on its past performance (Algorithm~\ref{multipleleadersnotimeout}). By $c_v(e)$ we denote the number of requests committed under leader $v$ in epoch $e$. Each leader is re-evaluated during the epoch-change. If a leader successfully committed all assigned requests in the preceding epoch, we double the number of requests this leader is given in the following epoch. Else, it is assigned the maximum number of requests it committed within the last $f+1$ epochs.
        
        \begin{figure}
            \centering
            \begin{tikzpicture}[scale =0.7]

            	\foreach \c in {2,4}{
              		\draw[dotted,thick] (2.2*\c,3)--++(0,-3);
             	}
             	 	
             	\node[] at (1.1*3,3.5) {\small leader prepare};
             	\node[] at (1.1*6,3.5) {\small backup prepare};
             	\node[] at (1.1*10,3.5) {\small commit};
             	 	
             	\foreach \a/\b/\c/\d in {
                    1.2/3/0.8/-1,
                    1.2/3/0.7/-2,
                    1.2/3/0.6/-3,
                    2.1/3/0.8/-1,
                    2.2/1/0.6/1,
                    2.0/0/0.9/2,
                    3.2/2/0.8/-1,
                    3.2/2/0.7/-2,
                    3.2/2/0.6/1,
                    4.1/3/0.8/-1,
                    4.1/1/0.7/1,
                    4.0/0/0.9/2,
                    5.2/2/0.8/-1,
                    5.2/2/0.7/-2,
                    5.2/2/0.7/1
                }{ \draw[cyan,thick,-{Stealth}] (2.2*\a,\b)--++(2.2*\c,\d); }
                
            	\foreach \row [count=\r] in {\small 0,\small 1,\small 2,\small 3} {
            		\draw[very thick](2*1.1,5-\r-1)node[left]{\row} -- ++(1.1*10,0);
              	}	
            \end{tikzpicture}
            \caption{Schematic of message flow for hanging requests. In this example, the primary is replica $0$, and the request falls into the bucket of replica $1$.}
            \label{fig:messageflowepochchange}
            \vspace{-0.5cm}
        \end{figure}
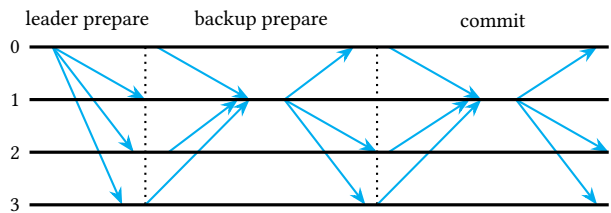

        \paragraph{Epoch-change timer.}\label{sec:timer}
            A replica sets an epoch-change timer~$T_{e}$ upon entering the epoch-change for epoch~$e+1$. By default, we configure the epoch-change timer~$T_{e}$ such that a correct primary can successfully finish the epoch-change after GST. If the timer expires without seeing a valid new-epoch message, the replica requests an epoch-change for epoch~$e+2$. If a replica has experienced at least $f$ unsuccessful consecutive epoch-changes previously, the replica doubles the timer's value. It continues to do so until it sees a valid new-epoch message. We only start doubling the timer after $f$ unsuccessful consecutive epoch-changes to avoid having $f$ byzantine primaries in a row, i.e., the maximum number of subsequent byzantine primaries possible, purposely increasing the timer value exponentially and, in turn, decreasing the system throughput significantly. As soon as replicas witness a successful epoch-change, they reduce $T_e$ to its default again.
      
        \paragraph{Assignment of requests}\label{sec:update}
            Finally, the number of requests assigned to each leader is updated during the epoch-change. We limit the number of requests that can be processed by each leader per epoch to assign the sequence numbers ahead of time and allow leaders to work independently of each other. 
        %
        %
            We assign sequence numbers to leaders according to their abilities. As soon as we see a leader outperforming their workload, we double the number of requests they are assigned in the following epoch. Additionally, leaders operating below their expected capabilities are allocated requests according to the highest potential demonstrated in the past $f+1$ rounds. By looking at the previous $f+1$ epochs, we ensure that there is at least one epoch with a correct primary in the leader set. In this epoch, the leader had the chance to display its capabilities. Thus, basing a leader's performance on the last $f+1$ rounds allows us to see its ability independent of the possible influence of byzantine primaries.

\section{Analysis}\label{sec:analysis}
    We show that \prot satisfies the properties specified in Section~\ref{sec:goals}.
    A detailed analysis can be found in Appendix~\ref{app:proofs}.
    \paragraph{Safety.}
    We prove \prot is safe under asynchrony. 
    \prot generalizes Linear-PBFT~\cite{gueta2018sbft}, which is an adaptation of PBFT~\cite{castro2002practical} that reduces its authenticator complexity during epoch operation. We thus rely on similar arguments to prove \prot's safety in Theorem~\ref{safety}.
    
    \paragraph{Liveness.}
    We show \prot makes progress after GST (Theorem~\ref{thm:liveness}).
    \prot's epoch-change uses the following techniques to ensure that correct replicas become synchronized (Definition~\ref{def:syn}) after GST:
    (1)~A replica in epoch~$e$ observing epoch-change messages from $f+1$ other replicas calling for any epoch(s) greater than $e$ issues an epoch-change message for the smallest such epoch.
    (2)~A replica only starts its epoch-change timer after receiving $2f$ other epoch-change messages, thus ensuring that at least $f + 1$ correct replicas have broadcasted an epoch-change message for the epoch (or higher). Hence, all correct replicas start their epoch-change timer for an epoch~$e'$ within at most $2$ message delay. After GST, this amounts to at most $2 \Delta$.
    (3)~Byzantine replicas are unable to impede progress by calling frequent epoch-changes, as an epoch-change will only happen if at least $f+1$ replicas call it. A byzantine primary can hinder the epoch-change from being successful. However, there can only be $f$ byzantine primaries in a row. 
        
    
    \paragraph{Efficiency.}
    To demonstrate that \prot is efficient, we start by analyzing the authenticator complexity for reaching consensus during an epoch. Like Linear-PBFT~\cite{gueta2018sbft}, using each leader as a collector for partial signatures in the backup prepare and commit phase, allows \prot to achieve linear complexity during epoch operation.
    We continue by calculating the authenticator complexity of an epoch-change. Intuitively speaking, we reduce PBFT's view-change complexity from $\Theta(n^3)$ to $\Theta(n^2)$ by employing threshold signatures.
    However, as \prot allows for $n$ simultaneous leaders, we obtain an authenticator complexity of $\Theta(n^3)$ as a consequence of sharing the same information for $n$ leaders during the epoch-change.
    Finally, we argue that after GST, there is sufficient progress by correct replicas to compensate for the high epoch-change cost (Theorem~\ref{amormult}).
    
    \paragraph{Optimistic Performance.}
    We assess \prot's optimistic performance, i.e., we theoretically evaluate its best-case throughput, assuming all replicas are correct and the network is synchronous.
    We further assume that the best-case throughput is limited by the available computing power of each replica -- mainly required for the computation and verification of cryptographic signatures -- and that the available computing power of each correct replica is the same. 
    
    In this model, we demonstrate that \prot achieves higher throughput than sequential-leader protocols by the means of leader parallelization.
    To show the speed-up gained through parallelization, we first analyze the \textit{optimistic epoch throughput} of \prot, i.e., the throughput of the system during stable networking conditions in the best-case scenario with $3f + 1$ correct replicas (Lemma~\ref{optimalthroughput}). 
    Later, we consider the repeated epoch changes and show that \prot's throughput is dominated by its authenticator complexity during the epochs. To that end, observe that for $C_{\min} \in \Omega(n^2)$, every epoch will incur an authenticator complexity of $\Omega(n^3)$ per replica and thus require $\Omega(n^3)$ time units.
    We show that after GST, an epoch-change under a correct primary requires $\Theta(n^2)$ time units (Lemma~\ref{lem:epoch-change-correct-primary}).
    We conclude our analysis by quantifying \prot's overall best-case throughput. Specifically, we prove that the speed-up gained by moving from a sequential-leader protocol to a parallel-leader protocol is proportional to the number of leaders (Theorem~\ref{speedup}).

    \paragraph{Byzantine-Resilient Performance.}
    While many BFT protocols present practical evaluations of their performance that  ignore byzantine adversarial behavior~\cite{castro2002practical,gueta2018sbft,yin2019hotstuff,stathakopoulou2019mir}, we provide a novel, theoretical byzantine-resilience guarantee. 
    %
    We first analyze the impact of byzantine replicas in an epoch under a correct primary. 
    We consider the replicas' roles as backups and leaders separately.
    On the one hand, for a byzantine leader, the optimal strategy is to leave as many requests hanging, while not making any progress (Lemma~\ref{lem:byzantine-leader}).
    On the other hand, as a backup, the optimal byzantine strategy is not helping other leaders to make progress (Lemma~\ref{lem:byzantine-backup}).
    In conclusion, we observe that byzantine replicas have little opportunity to reduce the throughput in epochs under a correct primary. Specifically, we show that after GST, the effective utilization under a correct primary is at least $\frac{8}{9}$ for $n \to \infty$ (Theorem~\ref{thm:byz-throughput-correct-primary}).

    Next, we discuss the potential strategies of a byzantine primary trying to stall the system. We fist show that under a byzantine primary, an epoch is either aborted quickly or $\Omega(n^3)$ new requests become committed (Lemma~\ref{byzprimary}).
    Then, we prove that  rotating primaries across epochs based on primary performance history reduces the control of the byzantine adversary on the system. In particular, byzantine primaries only have one turn as primary throughout any sliding window in a stable network.
    Combining all the above, we conclude that \prot's byzantine-resilient utilization is asymptotically $\frac{8}{9} \cdot \frac{g-f}{g}> \frac{16}{27}$ for $n \to \infty$ (Theorem~\ref{worstcase}), where $g$ is the fraction of byzantine primaries in the system's stable state, while simultaneously dictates how long it takes to get there after GST.

\section{Evaluation}\label{sec:eval}
In this section, we evaluate a prototype implementation of \prot with respect to performance, i.e., throughput and latency. We compare \prot to the state-of-the-art BFT protocol HotStuff~\cite{yin2019hotstuff}. 

\paragraph{Implementation.}
Our proof-of-concept implementation is \textit{directly based} on libhotstuff,\footnote{\url{https://github.com/hot-stuff/libhotstuff}} changing roughly 2000 lines of code while maintaining the general framework and experiment setup.
We implemented the basic functionality of our protocol as well as the epoch-change and watermarks. 
To adapt to our protocol while fairly comparing with HotStuff using the code from libhotstuff, we extended both implementations to support BLS threshold signatures.\footnote{\url{https://github.com/herumi/bls}}
Implementation details can be found in Appendix~\ref{app:implementation}.\footnote{We intend to open source the code.}

\paragraph{Setup.}
We evaluate \prot on Amazon EC2 using c5.4xlarge AWS cloud instances.
We run each replica on a single VM instance.
Each instance has 16 CPU cores powered by Intel Xeon Platinum 8000 processors clocked at up to 3.6 GHz, 32 GiB of RAM, and a maximum available TCP bandwidth of 10 Gigabits per second.

\paragraph{Methodology.}
We measure and report the average throughput and latency when all $n$ replicas operate correctly. 
We run the experiments for $n \in \{4, 7, 10, 16, 31, 61\}$. 

Our evaluation includes the epoch-change, so we measure the average throughput and latency over multiple epochs.
Specifically, for each experiment, we run both protocols over a period of three minutes and measure the average performance of \prot and HotStuff accordingly. We believe this is representative of \prot's performance as the prototype implementation is deterministic and therefore the difference between independent executions lies in the ordering of messages, which does neither affect the average throughput nor the average latency.

In our experiments, we divide the hash space into $n$ buckets resulting in one bucket per replica. Enough clients ($2n$) are instantiated such that no buckets are idle to accurately measure throughput.
For the latency measurement, we run a single client instance such that the system does not reach its throughput limit.
Each client initially generates and broadcasts 4 requests in parallel and issues a new request whenever one of its requests is answered.
Specifically, we run the libhotstuff client with its default settings, meaning that the payload of each request is 0.

As in the theoretical analysis, we do not employ a batching process for the evaluation. Hence, each block contains a single client request. For this reason, the expected throughput in practical deployments is much higher. Nevertheless, even with larger block size (e.g., 500 requests per block) we expect the bottleneck for \prot's performance to be the computation (authenticator complexity) as shown in our analysis, which is captured in our evaluation. To guarantee a fair comparison, we use the same settings for both HotStuff and \prot. We also distinguish the performance of the multi-threaded HotStuff (with 12 threads per replica) and the single-threaded one. Our implementation of \prot is currently single-threaded only but -- as we show below -- still outperforms multi-threaded HotStuff.

Table~\ref{tab:params} summarizes the configuration parameters for all experiments. When applicable, equal settings were chosen for HotStuff.

\begin{table}[hbt]
    \centering
        \begin{tabular}{lc}
            \toprule
            Configuration Parameter             &              Setting                           \\
            \midrule
            Requests per block                  &                $1$                             \\
            Threads per replica                 &                $1$                             \\
            Threads per client                  &                $4$                             \\
            Epoch timeout                       &               $30$s                            \\
            No progress timeout                 &                $2$s                            \\
            Blocks per checkpoint ($K$)         &                $50$                            \\
            Watermark window size ($2*K$)       &                $100$                           \\
            Initial epoch watermark bounds      &               $10000$                          \\
            \bottomrule
        \end{tabular}%
    \caption{%
        \prot configuration parameters used across all experiments.
    }
    \label{tab:params}
\end{table}
\vspace{-20pt}

\paragraph{Performance.}
Figure~\ref{fig:throughput} depicts a standard operation of \prot over 5 epochs and demonstrates its high throughput even when the batch size is one. As expected, the performance and throughput of our protocol stalls during an epoch-change.
However, the average throughput still remains significantly higher than HotStuff's throughput, especially when replicas increase, as illustrated in Figure~\ref{fig:average-throughput}. Specifically, \prot is at least 1.1$\times$ faster than multi-treaded HotStuff and at least 1.5$\times$ faster than single-threaded HotStuff. Figure~\ref{fig:average-latency} depicts the average latency of both \prot and HotStuff, showing that they scale similarly with the number of replicas. 
As latency expresses the time between a request being issued and committed, both protocols exhibit very fast finality for requests on average, even with many replicas.
In combination, Figure~\ref{fig:average-throughput} and Figure~\ref{fig:average-latency} demonstrate the significant scaling capabilities of \prot, and its competitiveness against state-of-the-art.

\begin{figure}[hb]
    \centering
    \includegraphics[width=\linewidth]{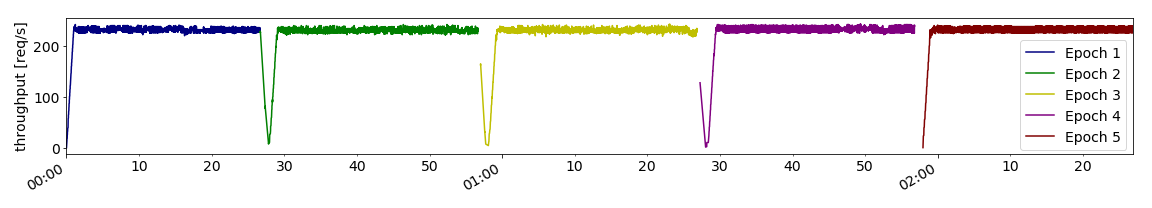}
    \vspace{-1cm}
    \caption{Throughput of \prot with $n = 4$ replicas over 5 epochs.}
    \label{fig:throughput}
    \begin{minipage}{0.5\textwidth}
        \centering
        \includegraphics[width=\linewidth]{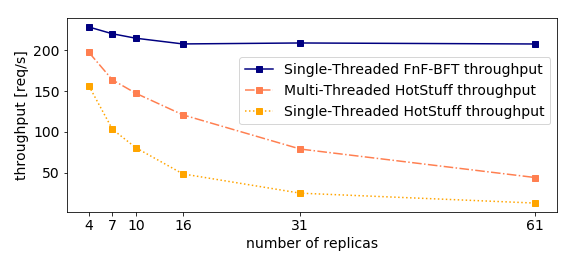}
        \vspace{-1cm}
        \caption{Average Throughput Comparison}
        \label{fig:average-throughput}
    \end{minipage}\hfill
    \begin{minipage}{0.5\textwidth}
        \centering
        \includegraphics[width=\linewidth]{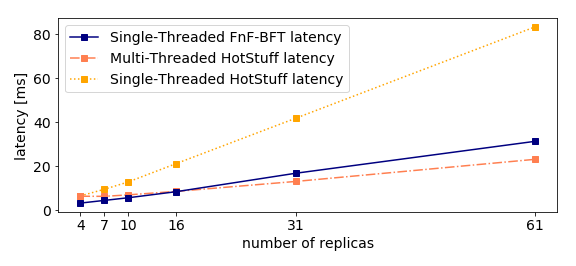}
        \vspace{-1cm}
        \caption{Average Latency Comparison}
        \label{fig:average-latency}
    \end{minipage}
\end{figure}

    
\newpage    
\bibliographystyle{ACM-Reference-Format}
\bibliography{citations}

\appendix

\section{Analysis}\label{app:proofs}

\subsection{Safety}
        \prot generalizes Linear-PBFT~\cite{gueta2018sbft}, which is an adaptation of PBFT~\cite{castro2002practical} that reduces its authenticator complexity during epoch operation. We thus rely on similar arguments to prove \prot's safety in Theorem~\ref{safety}.
        \begin{theorem}\label{safety}
            If any two correct replicas commit a request with the same sequence number, they both commit the same request.
        \end{theorem}
        \begin{proof}
            We start by showing that if  $\mathlist{\langle \text{prepared-certificate}, sn , e, \sigma(d) \rangle}$ exists, then $\mathlist{\langle \text{prepared-certificate}, sn , e, \sigma(d') \rangle}$ cannot exist for $d'\neq d$. Here, $d= h(sn \| e\| r)$ and $d'= h(sn \| e\| r')$. Further, we assume the probability of $r\neq r'$ and $d=d'$ to be negligible. The existence of $\mathlist{\langle \text{prepared-certificate}, sn , e, \sigma(d) \rangle}$ implies that at least $f+1$ correct replicas sent a pre-prepare message or a prepare message for the request~$r$ with digest~$d$ in epoch~$e$ with sequence number~$sn$. For $\mathlist{\langle \text{prepared-certificate}, sn , e, \sigma(d') \rangle}$ to exist, at least one of these correct replicas needs to have sent two conflicting prepare messages (pre-prepare messages in case it leads $sn$). This is a contradiction.
           
            Through the epoch-change protocol we further ensure that correct replicas agree on the sequence of requests that are committed locally in different epochs. The existence of $\mathlist{\langle \text{prepared-certificate}, sn , e, \sigma(d) \rangle}$ implies that $\mathlist{\langle \text{prepared-certificate}, sn , e', \sigma(d') \rangle}$ cannot exist for $d'\neq d$ and $e'>e$. Any correct replica only commits a request with sequence number~$sn$ in epoch~$e$ if it saw the corresponding commit-certificate. For a commit-certificate for request~$r$ with digest~$d$ and sequence number~$sn$ to exist a set~$R_1$ of at least $f+1$ correct replicas needs to have seen $\mathlist{\langle \text{prepared-certificate}, sn , e, \sigma(d) \rangle}$. A correct replica will only accept a pre-prepare message for epoch~$e'>e$ after having received a new-epoch message for epoch~$e'$. Any correct new-epoch message for epoch~$e'>e$ must contain epoch-change messages from a set~$R_2$ of at least $f+1$ correct replicas. As there are $2f+1$ correct replicas, $R_1$ and $R_2$ intersect in at least one correct replica~$u$. Replica~$u$'s epoch-change message ensures that information about request~$r$ being prepared in epoch~$e$ is propagated to subsequent epochs, unless $sn$ is already included in the stable checkpoint of its leader. In case the prepared-certificate is propagated to the subsequent epoch, a commit-certificate will potentially be propagated as well. If the new-epoch message only includes the prepared-certificate for $sn$, the protocol is redone for request~$r$ with the same sequence number~$sn$. In the two other cases, the replicas commit $sn$ locally upon seeing the new-epoch message and a correct replica will never accept a request with sequence number~$sn$ again.
        \end{proof}
        
    \subsection{Liveness}
        One cannot guarantee safety and liveness for deterministic BFT protocols in asynchrony~\cite{fischer1985impossibility}. We will, therefore, show that \prot eventually makes progress after GST. In other words, we consider a stable network when discussing liveness. Furthermore, we assume that after an extended period without progress, the time required for local computation in an epoch-change is negligible. Thus, we focus on analyzing the network delays for liveness.
        
        Similar to PBFT~\cite{castro2002practical}, \prot's epoch-change uses the following three techniques to ensure that correct replicas become synchronized (Definition~\ref{def:syn}) after GST.
        \begin{enumerate}
            \item \label{tec:broadcast-upon-f+1} A replica in epoch~$e$ observing epoch-change messages from $f+1$ other replicas calling for any epoch(s) greater than $e$ issues an epoch-change message for the smallest such epoch.
            \item \label{tec:start-view-change-timer} A replica only starts its epoch-change timer after receiving $2f$ other epoch-change messages, thus ensuring that at least $f + 1$ correct replicas have broadcasted an epoch-change message for the epoch (or higher). Hence, all correct replicas start their epoch-change timer for an epoch~$e'$ within at most $2$ message delay. After GST, this amounts to at most $2 \Delta$.
            \item \label{tec:byz-cannot-cause-view-change} Byzantine replicas are unable to impede progress by calling frequent epoch-changes, as an epoch-change will only happen if at least $f+1$ replicas call it. A byzantine primary can hinder the epoch-change from being successful. However, there can only be $f$ byzantine primaries in a row. 
        \end{enumerate}
        
        \begin{definition}\label{def:syn}
            Two replicas are called \textit{synchronized}, if they start their epoch-change timer for an epoch~$e$ within at most $2 \Delta$.
        \end{definition}
        
        \begin{lemma}
            \label{lem:synchronized-after-GST}
            After GST, correct replicas eventually become synchronized.
        \end{lemma}
        \begin{proof}
            Let $u$ be the first correct replica to start its epoch-change timer for epoch~$e$ at time~$t_0$. Following (\ref{tec:start-view-change-timer}), this implies that $u$ received at least $2f$ other epoch-change messages for epoch~$e$ (or higher). Of these $2f$ messages, at least $f$ originate from other correct replicas. Thus, together with its own epoch-change message, at least $f+1$ correct replicas broadcasted epoch-change messages by time~$t_0$. These $f+1$ epoch-change messages are seen by all correct replicas at the latest by time~$t_0 + \Delta$. Thus, according to (\ref{tec:broadcast-upon-f+1}), at time~$t_0 + \Delta$ all correct replicas broadcast an epoch-change message for epoch~$e$. Consequently, at time~$t_0 + 2\Delta$ all correct replicas have received at least $2f$ other epoch-change messages and will start the timer for epoch~$e$.
        \end{proof}
        
        \begin{lemma}
            \label{lem:same-epoch}
            After GST, all correct replicas will be in the same epoch long enough for a correct leader to make progress.
        \end{lemma}
        \begin{proof}
            From Lemma~\ref{lem:synchronized-after-GST}, we conclude that after GST, all correct replicas will eventually enter the same epoch if the epoch-change timer is sufficiently large. Once the correct replicas are synchronized in their epoch, the duration needed for a correct leader to commit a request is bounded. Note that all correct replicas will be in the same epoch for a sufficiently long time as the timers are configured accordingly. Additionally, byzantine replicas are unable to impede progress by calling frequent epoch-changes, according to (\ref{tec:byz-cannot-cause-view-change}).
        \end{proof}

        \begin{theorem}\label{thm:liveness}
            If a correct client~$c$ broadcasts request~$r$, then every correct replica eventually commits $r$.
        \end{theorem}
        \begin{proof}
            Following Lemmas~\ref{lem:synchronized-after-GST} and \ref{lem:same-epoch}, we know that all correct replicas will eventually be in the same epoch after GST. Hence, in any epoch with a correct primary, the system will make progress. Note that a correct client will not issue invalid requests. It remains to show that an epoch with a correct primary and a correct leader assigned to hash bucket $h(c)$ will occur. We note that this is given by the bucket rotation, which ensures that a correct leader repeatedly serves each bucket in a correct primary epoch. 
        \end{proof}
    \subsection{Efficiency}
        To demonstrate that \prot is efficient, we start by analyzing the authenticator complexity for reaching consensus during an epoch. Like Linear-PBFT~\cite{gueta2018sbft}, using each leader as a collector for partial signatures in the backup prepare and commit phase, allows \prot to achieve linear complexity during epoch operation.
        \begin{lemma}
            \label{lem:complexity-during-epoch}
            The authenticator complexity for committing a request during an epoch is $\Theta(n)$.
        \end{lemma}
        \begin{proof}
            During the leader prepare phase, the authenticator complexity is at most $n$. The primary computes its signature to attach it to the pre-prepare message. This signature is verified by no more than $n-1$ replicas.
            
            Furthermore, the backup prepare and commit phase's authenticator complexity is less than $3n$ each. Initially, at most $n-1$ backups, compute their partial signature and send it to the leader, who, in turn, verifies $2f$ of these signatures. The leader then computes its partial signature, as well as computing the combined signature. Upon receiving the combined signature, the $n-1$ backups need to verify the signature.
            
            Overall, the authenticator complexity committing a request during an epoch is thus at most $7n + o(n)\in \Theta(n)$.
        \end{proof}
        
        We continue by calculating the authenticator complexity of an epoch-change. Intuitively speaking, we reduce PBFT's view-change complexity from $\Theta(n^3)$ to $\Theta(n^2)$ by employing threshold signatures.
        However, as \prot allows for $n$ simultaneous leaders, we obtain an authenticator complexity of $\Theta(n^3)$ as a consequence of sharing the same information for $n$ leaders during the epoch-change.
        \begin{lemma}\label{lem:complexity-epoch-change}
            The authenticator complexity of an epoch-change is $\Theta(n^3)$.
        \end{lemma}
        \begin{proof}
            The epoch-change for epoch~$e+1$ is initiated by replicas sending epoch-change messages to the primary of epoch~$e+1$. Each epoch-change message holds $n$ authenticators for each leader's last checkpoint-certificates. As there are at most $2k$ hanging requests per leader, a further $\mathcal{O}(n)$ authenticators for prepared- and commit-certificates of the open requests per leader are included in the message. Additionally, the sending replica also includes its signature. Each replica newly computes its signature to sign the epoch-change message, the remaining authenticators are already available and do not need to be created by the replicas. Thus, a total of no more than $n$ authenticators are computed for the epoch-change messages. We also note epoch-change message contains $\Theta(n)$ authenticators. Therefore, the number of authenticators received by each replica is $ \Theta(n^2)$.
            
            After the collection of $2f+1$ epoch-change messages, the primary performs a classical 3-phase reliable broadcast. The primary broadcasts the same signed message to start the classical 3-phase reliable broadcast. While the primary computes $1$ signature, at most $n-1$ replicas verify this signature. In the two subsequent rounds of all-to-all communication, each participating replica computes $1$ and verifies $2f$ signatures. Thereby, the authenticator complexity of each round of all-to-all communication is at most $(2f+1)\cdot n$. Thus, the authenticator complexity of the 3-phase reliable broadcast is bounded by $(4f+3)\cdot n\in \Theta (n^2)$.
                
            After successfully performing the reliable broadcast, the primary sends out a new-epoch message to every replica in the network. The new-epoch message contains the epoch-change messages held by the primary and the required pre-prepare messages for open requests. There are $\mathcal{O}(n)$ such pre-prepare messages, all signed by the primary. Finally, each new-epoch message is signed by the primary. Thus, the authenticator complexity  of the new-epoch message is  $\Theta(n^2)$. However, suppose a replica has previously received and verified an epoch-change from replica~$u$ whose epoch-change message is included in the new-epoch message. In that case, the replica no longer has to check the authenticators in $u$'s epoch-change message again. For the complexity analysis, it does not matter when the replicas verify the signature. We assume that all replicas verify the signatures contained in the epoch-change messages before receiving the new-epoch messages. Thus, the replicas only need to verify the $\mathcal{O}(n)$ new authenticators contained in the new-epoch message. 
            
            Overall, the authenticator complexity of the the epoch-change is at most $\Theta( n^3)$.
        \end{proof}
        
        Finally, we argue that after GST, there is sufficient progress by correct replicas to compensate for the high epoch-change cost.

        \begin{theorem}\label{amormult}
            After GST, the amortized authenticator complexity of committing a request is $\Theta(n)$.
        \end{theorem}
        \begin{proof}
            
            To find the amortized authenticator complexity of committing a request, we consider an epoch and the following epoch-change.
            After GST, the authenticator complexity of committing a request for a correct leader is $\Theta(n)$. The timeout value is set such that a correct worst-case leader creates at least $C_{\min}$ requests in each epoch initiated by a correct primary. Thus, there are $\Theta (n)$ correct replicas, each committing $C_{\min}$ requests. By setting $C_{\min} \in \Omega (n^2)$, we guarantee that at least $\Omega (n^3)$ requests are created during an epoch given a correct primary. 
            
            Byzantine primaries can ensure that no progress is made in epochs they initiate, by failing to share the new-epoch message with correct replicas. However, at most, a constant fraction of epochs lies in the responsibility of byzantine primaries. We conclude that, on average, $\Omega(n^3)$ requests are created during an epoch.
                    
            Following Lemma~\ref{lem:complexity-epoch-change}, the authenticator complexity of an epoch-change is $\Theta(n^3)$. Note that the epoch-change timeout~$ T_e$ is set so that correct primaries can successfully finish the epoch-change after GST. Not every epoch-change will be successful immediately, as byzantine primaries might cause unsuccessful epoch-changes. Specifically, byzantine primaries can purposefully summon an unsuccessful epoch-change to decrease efficiency.
            
            In case of an unsuccessful epoch-change, replicas initiate another epoch-change -- and continue doing so -- until a successful epoch-change occurs. However, we only need to start $\mathcal{O}(1)$ epoch-changes on average to be successful after GST, as the primary rotation ensures that at least a constant fraction of primaries is correct. Hence, the average cost required to reach a successful epoch-change is $\Theta(n^3)$.
            
            We find the amortized request creation cost by adding the request creation cost to the ratio between the cost of a successful epoch-change and the number of requests created in an epoch, that is, $\Theta(n) + \frac{\Theta(n^3)}{\Omega(n^3)}= \Theta(n).$
        \end{proof}
    

    \subsection{Optimistic Performance}
      
        Throughout this section, we make the following optimistic assumptions: all replicas are considered correct, and the network is stable and synchronous. We employ this model to assess the optimistic performance of \prot, i.e., theoretically evaluating its best-case throughput. Note that this scenario is motivated by practical applications, as one would hope to have functioning hardware at hand, at least initially.
        Additionally, we assume that the best-case throughput is limited by the available computing power of each replica -- predominantly required for the computation and verification of cryptographic signatures. We further assume that the available computing power of each correct replica is the same, which we believe is realistic as the same hardware will often be employed for each replica in practice. Without loss of generality, each leader can compute/verify one authenticator per time unit. As \textit{throughput}, we define the number of requests committed by the system per time unit.
        Finally, we assume that replicas only verify the authenticators of relevant messages. For example, a leader receiving $3f$ prepare messages for a request will only verify $2f$ authenticators. Similarly, pre-prepare messages outside the leaders' watermarks will not be processed by backups. 
        Note that we will carry all assumptions into Section~\ref{sec:byzantine-resilient-performance}. There they will, however, only apply to correct replicas. 
    
        \subsubsection{Sequential-Leader Protocols}
            We claim that \prot achieves higher throughput than sequential-leader protocols by the means of leader parallelization. To support this claim, we compare \prot's throughput to that of a generic sequential-leader protocol. The generic sequential-leader protocol will serve as an asymptotic characterization of many state-of-the-art sequential-leader protocols~\cite{castro2002practical,gueta2018sbft,yin2019hotstuff}.
            
            A \textit{sequential-leader protocol} is characterized by having a unique leader at any point in time. Throughout its reign, the leader is responsible for serving all client requests. Depending on the protocol, the leader is rotated repeatedly or only upon failure.
            \begin{lemma}
                \label{lem:sequential-leader-throughput}
                A sequential-leader protocol requires at least $\Omega(n)$ time units to process a client request.
            \end{lemma}
            \begin{proof}
                In sequential-leader protocols, a unique replica is responsible for serving all client requests at any point in time. This replica must verify $\Omega(n)$ signatures to commit a request while no other replica leads requests simultaneously. Thus, a sequential-leader protocol requires  $\Omega(n)$ time units to process a request.
            \end{proof}
            \begin{figure}[htpb]
                \centering
                    
                \begin{tikzpicture}[shorten >=1pt,auto,node distance=3cm,semithick,scale = 0.7, state/.style={circle,draw, minimum size=0.59cm}]
                            	
                    \node[state] (A) at (0,0) {$1$};	
                    \node[state] (B) at (1.2,0) {$2$};	
                    \node[state] (C) at (2.4,0) {$3$};
                    \node[state] (D) at (3.6,0) {$4$};
                    \node[state] (E) at (4.8,0) {$1$};
                    \node[state] (F) at (6,0) {$2$};
                    \node[state] (G) at (7.2,0) {$3$};	
                    \node[state] (H) at (8.4,0) {$4$};
                    \node[state] (I) at (9.6,0) {$1$};
                    \node[state] (J) at (10.8,0) {$2$};	
                    \node at (11.7,0) [circle,fill,inner sep=1.2pt]{};
                    \node at (12.0,0) [circle,fill,inner sep=1.2pt]{};
                    \node at (12.3,0) [circle,fill,inner sep=1.2pt]{};
                                
                    \draw (0.6,-0.6) -- (0.6,0.6);
                    \draw (1.8,-0.6) -- (1.8,0.6);
                    \draw (3,-0.6) -- (3,0.6);
                    \draw (4.2,-0.6) -- (4.2,0.6);
                    \draw (5.4,-0.6) -- (5.4,0.6);
                    \draw (6.6,-0.6) -- (6.6,0.6);
                    \draw (7.8,-0.6) -- (7.8,0.6);
                    \draw (9,-0.6) -- (9,0.6);
                    \draw (10.2,-0.6) -- (10.2,0.6);
                    \draw (11.4,-0.6) -- (11.4,0.6);
                \end{tikzpicture}
                    
                \caption{Sequential leader example with four leaders. Throughout its reign a sequential leader is responsible for serving all client requests. Leader changes are indicated by vertical lines.}
                \label{fig:sequential}
            \end{figure}
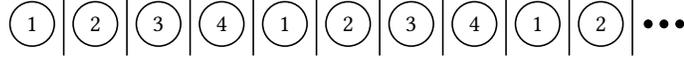
            
        \subsubsection{\prot Epoch}
            With \prot, we propose a \textit{parallel-leader protocol} (cf.\ Figure~\ref{fig:parallel}) that divides client requests into $m\cdot n$ independent hash buckets. Each hash bucket is assigned to a unique leader at any time. The hash buckets are rotated between leaders across epochs to ensure liveness (cf.\ Section~\ref{sec:hash}). Within an epoch, a leader is only responsible for committing client requests from its assigned hash bucket(s). Overall, this parallelization leads to a significant speed-up.
            
            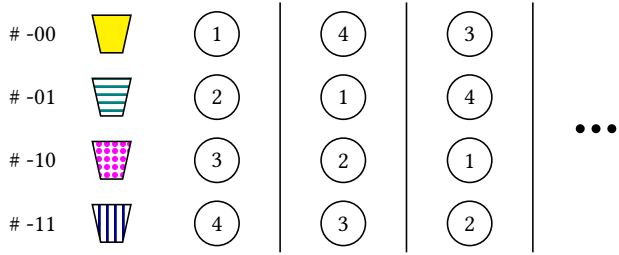
\begin{figure}[htpb]   
                \centering
                \begin{tikzpicture}[shorten >=1pt,auto,node distance=3cm,semithick,scale = 0.7, state/.style={circle,draw, minimum size=0.59cm},mytrap/.style={trapezium, trapezium angle=102.5, draw,inner xsep=0pt,outer sep=0pt,minimum height=5mm, text width=#1}]
                    \definecolor{navy}{rgb}{0,0,0.5}
                    \node[state] (A1) at (0,0) {$1$};
                    \node [mytrap=3.5 pt,fill =yellow] at (-2,0) {};
                    \node[] (A21) at (-3.5,0) {\# -$00$};
                    \node[state] (B1) at (0,-1.2) {$2$};
                    \node [mytrap=3.5 pt,pattern=horizontal, hatch distance=4pt, hatch thickness = 2pt, pattern color=teal] at (-2,-1.2) {};
                    \node[] (B21) at (-3.5,-1.2) {\# -$01$};
                    \node[state] (C1) at (0,-2.4) {$3$};
                    \node [mytrap=3.5 pt,pattern=mydots, pattern color=magenta] at (-2,-2.4) {};
                    \node[] (C21) at (-3.5,-2.4) {\# -$10$};
                    \node[state] (D1) at (0,-3.6) {$4$};
                    \node [mytrap=3.5 pt,pattern=vertical, hatch distance=4pt, hatch thickness = 2pt, pattern color=navy] at (-2,-3.6) {};
                    \node[] (D21) at (-3.5,-3.6) {\# -$11$};
                    \draw (1.2,0.6) -- (1.2,-4.2);
                    \node[state] (A2) at (2.4,0) {$4$};	
                    \node[state] (B2) at (2.4,-1.2) {$1$};
                    \node[state] (C2) at (2.4,-2.4) {$2$};
                    \node[state] (D2) at (2.4,-3.6) {$3$};
                    \draw (3.6,0.6) -- (3.6,-4.2);
                    \node[state] (A3) at (4.8,0) {$3$};	
                    \node[state] (B3) at (4.8,-1.2) {$4$};
                    \node[state] (C3) at (4.8,-2.4) {$1$};
                    \node[state] (D3) at (4.8,-3.6) {$2$};
                    \draw (6,0.6) -- (6,-4.2);
                    \node at (6.9,-1.8) [circle,fill,inner sep=1.2pt]{};
                    \node at (7.2,-1.8) [circle,fill,inner sep=1.2pt]{};
                    \node at (7.5,-1.8) [circle,fill,inner sep=1.2pt]{};
                \end{tikzpicture}
                \caption{Parallel leader example with four leaders and four hash buckets. In each epoch, leaders are only responsible for serving client requests in their hash bucket. Epoch-changes are indicated by vertical lines.}
                \label{fig:parallel}
            \end{figure}
            
            To show the speed-up gained through parallelization, we first analyze the \textit{optimistic epoch throughput} of \prot, i.e., the throughput of the system during stable networking conditions in the best-case scenario with $3f + 1$ correct replicas. Furthermore, we assume the number of requests included in a checkpoint to be sufficiently large, such that no leader must ever stall when waiting for a checkpoint to be created. We analyze the effects of epoch-changes and compute the overall best-case throughput of \prot in the aforementioned optimistic setting.

            \begin{lemma}\label{optimalthroughput}
                After GST, the best-case epoch throughput with $3f+1$ correct replicas is $\dfrac{k\cdot (3f+1)}{k\cdot (19f+3)+(8f+2)}.$
            \end{lemma}
            \begin{proof}
                In the optimistic setting, all epochs are initiated by correct primaries, and thus all replicas will be synchronized after GST.

                In \prot, $n$ leaders work on client requests simultaneously. As in sequential-leader protocols, each leader needs to verify at least $\mathcal{O}(n)$ signatures to commit a request. A leader needs to compute $3$ and verify $4f$ authenticators precisely to commit a request it proposes during epoch operation. Thus, leaders need to process a total of $4f+3\in \Theta(n)$ signatures to commit a request. With the help of threshold signatures, backups involved in committing a request only need to compute $2$ and verify $3$ authenticators. We follow that a total of $4f+3 +5\cdot 3f=19f+3$ authenticators are computed/verified by a replica for one of its own requests and $3f$ requests of other leaders.
                
                After GST, each correct leader~$v$ will quickly converge to a $C_v$ such that it will make progress for the entire epoch-time, hence, working at its full potential. We achieve this by rapidly increasing the number of requests assigned to each leader outperforming its assignment and never decreasing the assignment below what the replica recently managed.
                
                Checkpoints are created every $k$ requests and add to the computational load. A leader verifies and computes a total of $2f+2$ messages to create a checkpoint, and the backups are required to compute $1$ partial signature and verify $1$ threshold signature. The authenticator cost of creating $3f+1$ checkpoints, one for each leader, is, therefore, $8f+2$ per replica.
                
                Thus, the best-case throughput of the system is $\dfrac{ k\cdot (3f+1)}{k\cdot (19f+3)+(8f+2)}.$
            \end{proof}
            Note that it would have been sufficient to show that the epoch throughput is $\Omega(1)$ per time unit, but this more precise formula will be required in Section~\ref{sec:byzantine-resilient-performance}. Additionally, we would like to point out that the choice of $k$ does not influence the best-case throughput asymptotically.
            
        \subsubsection{\prot Epoch-Change}
            As \prot employs bounded-length epochs, repeated epoch-changes have to be considered. In the following, we will show that \prot's throughput is dominated by its authenticator complexity during the epochs. To that end, observe that for $C_{\min} \in \Omega(n^2)$, every epoch will incur an authenticator complexity of $\Omega(n^3)$ per replica and thus require $\Omega(n^3)$ time units.
            
            \begin{lemma}\label{lem:epoch-change-correct-primary}
                After GST, an epoch-change under a correct primary requires $\Theta(n^2)$ time units.
            \end{lemma}
            \begin{proof}
                Following Lemma~\ref{lem:complexity-epoch-change}, the number of authenticators computed and verified by each replica for all epoch-change messages is $\Theta(n^2)$. Each replica also processes $\Theta(n)$ signatures during the reliable broadcast, and $\mathcal{O}(n)$ signatures for the new-epoch messages. Overall, each replica thus processes $\Theta(n^2)$ authenticators during the epoch-change. Subsequently, this implies that the epoch-change requires $\Theta(n^2)$ time units, as we require only a constant number of message delays to initiate and complete the epoch-change protocol. Recall that we assume the throughput to be limited by the available computing power of each replica.
            \end{proof}
            
            Theoretically, one could set $C_{\min}$ even higher such that the time the system spends with epoch-changes becomes negligible. However, there is a trade-off for practical reasons: increasing $C_{\min}$ will naturally increase the minimal epoch-length, allowing a byzantine primary to slow down the system for a longer time stretch. Note that the guarantee for byzantine-resilient performance (cf.\ Section~\ref{sec:byzantine-resilient-performance}) will still hold.
            
        \subsubsection{\prot Optimistic Performance}
            Ultimately, it remains to quantify \prot's overall best-case throughput.
            \begin{lemma}
                \label{lem:optimistic-throughput}
                After GST, and assuming all replicas are correct, \prot requires $\mathcal{O}(n)$ time units to process $n$ client requests on average.
            \end{lemma}
            \begin{proof}
                Under a correct primary, each correct leader will commit at least $C_{\min} \in \Omega(n^2)$ requests after GST. Hence, \prot will spend at least $\Omega(n^3)$ time units in an epoch, while only requiring $\Theta(n^2)$ time units for an epoch-change (Lemma~\ref{lem:epoch-change-correct-primary}).
                Thus, following Lemma~\ref{optimalthroughput}, \prot requires an average of $\mathcal{O}(n)$ time units to process $n$ client requests.
            \end{proof}
            
            Following Lemmas~\ref{lem:sequential-leader-throughput} and \ref{lem:optimistic-throughput}, the speed-up gained by moving from a sequential-leader protocol to a parallel-leader protocol is proportional to the number of leaders.
            \begin{theorem}\label{speedup}
                If the throughput is limited by the (equally) available computing power at each replica, the speed-up for equally splitting requests between $n$ parallel leaders over a sequential-leader protocol is at least $\Omega(n)$.
            \end{theorem}

    \subsection{Byzantine-Resilient Performance}
    \label{sec:byzantine-resilient-performance}
        While many BFT protocols present practical evaluations of their performance that completely ignore byzantine adversarial behavior~\cite{castro2002practical,gueta2018sbft,yin2019hotstuff,stathakopoulou2019mir}, we provide a novel, theory-based byzantine-resilience guarantee. We first analyze the impact of byzantine replicas in an epoch under a correct primary. Next, we discuss the potential strategies of a byzantine primary trying to stall the system. And finally, we conflate our observations into a concise statement.
        
        \subsubsection{Correct Primary Throughput}
            To gain insight into the byzantine-resilient performance, we analyze the optimal byzantine strategy. In epochs led by correct primaries, we will consider their roles as backups and leaders separately.
            On the one hand, for a byzantine leader, the optimal strategy is to leave as many requests hanging, while not making any progress (Lemma~\ref{lem:byzantine-leader}).
            \begin{lemma}
                \label{lem:byzantine-leader}
                After GST and under a correct primary, the optimal strategy for a byzantine leader is to leave $2k$ client requests hanging and commit no request.
            \end{lemma}
            \begin{proof}
                Correct replicas will be synchronized as a correct primary initiates the epoch. Thus, byzantine replicas' participation is not required for correct leaders to make progress.
                
                A byzantine leader can follow the protocol accurately (at any chosen speed), send messages that do not comply with the protocol, or remain unresponsive. 
                
                If following the protocol, a byzantine leader can open at most $2k$ client requests simultaneously as all further prepare messages would be discarded. Leaving the maximum possible number of requests hanging achieves a throughput reduction as it increases the number of authenticators shared during the epoch and the epoch-change. Hence, byzantine leaders leave the maximum number of requests hanging.
                
                While byzantine replicas cannot hinder correct leaders from committing requests, committing any request can only benefit the throughput of \prot. To that end, note that after GST, each correct leader~$v$ will converge to a $C_v$ such that it will make progress during the entire epoch-time; hence, prolonging the epoch-time is impossible. The optimal strategy for byzantine leaders is thus to stall progress on their assigned hash buckets.
                
                Finally, note that we assume the threshold signature scheme to be robust and can, therefore, discard any irrelevant message efficiently.
            \end{proof}
            
            On the other hand, as a backup, the optimal byzantine strategy is not helping other leaders to make progress (Lemma~\ref{lem:byzantine-backup}).
            \begin{lemma}
                \label{lem:byzantine-backup}
                Under a correct primary, the optimal strategy for a byzantine backup is to remain unresponsive.
            \end{lemma}
            \begin{proof}
                Byzantine participation in the protocol can only benefit the correct leaders' throughput as they can simply ignore invalid messages. Any authenticators received in excess messages will not be verified and thus do not reduce the system throughput.

            \end{proof}
            
            In conclusion, we observe that byzantine replicas have little opportunity to reduce the throughput in epochs under a correct primary.
            \begin{theorem}
                \label{thm:byz-throughput-correct-primary}
                After GST, the effective utilization under a correct primary is at least $\frac{8}{9}$ for $n \to \infty$.
            \end{theorem}
            \begin{proof}
                Moving from the best-case scenario with $3f+1$ correct leaders to only $2f+1$ correct leaders, each correct leader still processes $4f+3$ authenticators per request, and $5$ authenticators for each request of other leaders. We know from Lemma~\ref{lem:byzantine-leader} that only the $2f+1$ correct replicas are committing requests and creating checkpoints throughout the epoch. The authenticator cost of creating $2f+1$ checkpoints, one for each correct leader, becomes $6f+2$ per replica.
                
                Byzantine leaders can open at most $2k$ new requests in an epoch. Each hanging request is seen at most twice by correct replicas without becoming committed. Thus, each correct replica processes no more than $8k$ authenticators for requests purposefully left hanging by a byzantine replica in an epoch. Thus, the utilization is reduced at most by a factor $\left(1-\frac{8 k f}{T}\right)$, where $T$ is the maximal epoch length. While epochs can finish earlier, this will not happen after GST as soon as each correct leader~$v$ works at its capacity~$C_v$.
                
                Hence, the byzantine-resilient epoch throughput becomes $$\frac{k\cdot (2f+1)}{k\cdot (14f+3)+(6f+2)} \cdot \left(1-\frac{8  kf}{T}\right).$$
                
                By comparing this to the best-case epoch throughput from Lemma~\ref{optimalthroughput}, we obtain a maximal throughput reduction of
                \begin{align*}
                    \frac{ (2f+1)(k \cdot (19f+3)+(8f+2))}{(3f+1)(k \cdot (14f+3)+(6f+2))}\cdot \left(1-\frac{8  k  f}{ T}\right).
                \end{align*}
                Observe that the first term decreases and approaches $\frac{8}{9}$ for $n \to \infty$:
                \begin{align*}
                    \frac{(2f+1)(k \cdot (19f+3)+(8f+2))}{(3f+1)(k \cdot (14f+3)+(6f+2))}\stackrel{n\rightarrow \infty}{=}\frac{16 + 38 k}{18  + 42  k }\geq \frac{8}{9}.
                \end{align*}
                
                We follow that the epoch time is $T \in \Omega (n^3)$, as we set $C_{\min}\in \Omega(n^2)$ and each leader requires $\Omega(n)$ time units to commit one of its requests. Additionally, we know that $8 k  f\in \mathcal{O}(n)$, and thus: $\left(1-\dfrac{8  k  f}{ T}\right)\stackrel{n\rightarrow \infty}{=}1.$
                
                In the limit $n\rightarrow \infty$, the throughput reduction byzantine replicas can impose on the system during a synchronized epoch is therefore bounded by a factor $\dfrac{8}{9}$.
            \end{proof}
            
        \subsubsection{Byzantine Primary Throughput}
            A byzantine primary, evidently, aims to perform the epoch-change as slow as possible. Furthermore, a byzantine primary can impede progress in its assigned epoch entirely, e.g., by remaining unresponsive. We observe that there are two main byzantine strategies to be considered.
            \begin{lemma}\label{byzprimary}
                Under a byzantine primary, an epoch is either aborted quickly or $\Omega(n^3)$ new requests become committed.
            \end{lemma}
            \begin{proof}
                A byzantine adversary controlling the primary of an epoch has three options. Following the protocol and initiating the epoch for all $2f+1$ correct replicas will ensure high throughput and is thus not optimal. Alternatively, initiating the epoch for $s \in [f+1, 2f]$ correct replicas will allow the byzantine adversary to control the progress made in the epoch, as no correct leader can make progress without a response from at least one byzantine replica. However, slow progress can only be maintained as long as at least $2f + 1$ leaders continuously make progress. By setting the no-progress timeout~$T_p\in \Theta(T/C_{\min})$, $\Omega(n^3)$ new requests per epoch can be guaranteed. In all other scenarios, the epoch will be aborted after at most one epoch-change timeout~$T_e$, the initial message transmission time~$5\Delta$, and one no-progress timeout~$T_p$.
                        
                Note that we do not increase the epoch-change timer~$T_e$ for $f$ unsuccessful epoch-changes in a row. In doing so, we prevent $f$ consecutive byzantine primaries from increasing the epoch-change timer exponentially; thus potentially reducing the system throughput significantly.
            \end{proof}
            
        \subsubsection{\prot Primaries}
            We rotate primaries across epochs based on primary performance history to reduce the control of the byzantine adversary on the system. 
            \begin{lemma}
                \label{lem:primary}
                After a sufficiently long stable time period, the performance of a byzantine primary can only drop below the performance of the worst correct primary once throughout the sliding window. 
            \end{lemma}
            \begin{proof}
                The network is considered stable for a sufficiently long time when all leaders work at their capacity limit, i.e., the number of requests they are assigned in an epoch matches their capacity, and primaries have subsequently been explored once. As soon as all leaders are working at their capacity limit, we observe the representative performance of all correct primaries, at least.
                
                \prot repeatedly cycles through the $2f+1$ best primaries. A primary's performance is based on its last turn as primary. Consequently, a primary is removed from the rotation as soon as its performance drops below one of the $f$ remaining primaries. We conclude that a byzantine primary will only be nominated beyond its single exploration throughout the sliding window if its performance matches at least the performance of the worst correct primary.
            \end{proof}
            
            As its successor determines a primary's performance, the successor can influence the performance slightly. However, this is bounded by the number of open requests -- $\mathcal{O}(n)$ many -- which we consider being well within natural performance variations, as $\Omega(n^3)$ requests are created in an epoch under a correct primary. Thus, we will disregard possible performance degradation originating at the succeeding primary.
            
            From Lemma~\ref{lem:primary}, we easily see that the optimal strategy for a byzantine primary is to act according to Lemma~\ref{byzprimary} -- performing better would only help the system. In a stable network, byzantine primaries will thus only have one turn as primary throughout any sliding window. In the following, we consider a primary to be behaving byzantine if it performs worse than all correct primaries.
            
            \begin{theorem}
                \label{thm:proportionbyzprimary}
                After the system has been in stability for a sufficiently long time period, the fraction of byzantine behaving primaries is $\frac{f}{g}$. 
            \end{theorem}
            \begin{proof}
                Following from Lemma~\ref{lem:primary}, we know that a primary can only behave byzantine once in the sliding window.
                There are a total of $g$ epochs in a sliding window, and the $f$ byzantine replicas in the network can only act byzantine in one epoch included in the sliding window. We see that the fraction of byzantine behaving primaries is $\frac{f}{g}$.
            \end{proof}
            
            The configuration parameter $g$ determines the fraction of byzantine primaries in the system's stable state, while simultaneously dictating how long it takes to get there after GST. Setting $g$ to a small value ensures that the system quickly recovers from asynchrony. On the other hand, setting $g$ to larger values provides near-optimal behavior once the system is operating at its optimum.
            
        \subsubsection{\prot Byzantine-Resilient Performance}
        
            Combining the byzantine strategies from Theorem~\ref{thm:byz-throughput-correct-primary}, Lemma~\ref{byzprimary} and Theorem~\ref{thm:proportionbyzprimary}, we obtain the following.
            \begin{theorem}\label{worstcase}
                After GST, the effective utilization is asymptotically $\frac{8}{9} \cdot \frac{g-f}{g}$ for $n \to \infty$.
            \end{theorem} 
            \begin{proof}
                To estimate the effective utilization, we only consider the throughput within epochs. That is because the time spent in correct epochs dominates the time for epoch-changes, as well as the time for failed epoch-changes under byzantine primaries, as the number of replicas increases (Lemma~\ref{lem:epoch-change-correct-primary}).
                Without loss of generality, we consider no progress to be made in byzantine primary epochs. We make this assumption, as we cannot guarantee asymptotically significant throughput. 
                From Theorem~\ref{thm:byz-throughput-correct-primary}, we know that in an epoch initiated by a correct primary, the byzantine-resilient effective utilization is at least $\frac{8}{9}$ for $n \to \infty$. Further, at least $\frac{g-f}{g}$ of the epochs are led by correct primaries after a sufficiently long time period in stability and thus obey this bound (Theorem~\ref{thm:proportionbyzprimary}). 
                In the limit for $n \to \infty$ the effective utilization is $\frac{8}{9} \cdot\frac{g-f}{g}$.
            \end{proof}

\section{Implementation}\label{app:implementation}
We first provide an overview of \prot's messages and their format, and then describe \prot's behavior through epoch operation and epoch change.

\subsection{Messages and Data Structures}\label{cha:messages_overview}
The following is a non-exhaustive list of data structures and message formats used in \prot.


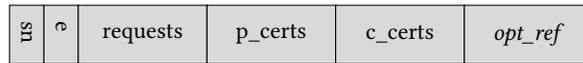
\begin{figure}[H]
    \centering
        \begin{tikzpicture}[node distance=-\pgflinewidth]
\node[short] (a) {\rotatebox{270}{sn}};
\node[short, right=of a] (b) {\rotatebox{270}{e}};
\lnode{b}{c}{requests};
\lnode{c}{d}{p\_certs};
\lnode{d}{e}{c\_certs};
\lnode{e} {f} {\textit{opt\_ref}};


\end{tikzpicture}
    \caption{%
        Data structure format of \PACK (known as block). A \PACK consists of a sequence number $sn$, an epoch number $e$, a list of requests $requests$ (to allow for batch-processing of requests), prepared certificates $p\_certs$ of replicas having prepared this \PACK, commit certificates $c\_certs$ of replicas having sent a \COMMIT message and an optional back-reference $opt\_ref$, if a \PACK refers to another \PACK (this is mainly used for hanging requests during epoch-changes).
    }
    \label{fig:pack}
\end{figure}
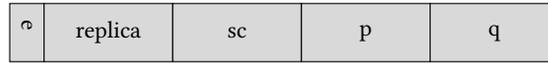
\begin{figure}[H]
    \centering

        \begin{tikzpicture}[node distance=-\pgflinewidth]

\node[short] (a) {\rotatebox{270}{e}};
\node[long,right=of a, label=center:replica] (b) {};
\node[long,right=of b, label=center:sc] (c) {};
\node[long,right=of c, label=center:p] (d) {};
\node[long,right=of d, label=center:q] (e) {};

\end{tikzpicture}

    \caption{%
        Data structure format of \ECO. An epoch-change object consists of an epoch number $e$, an unique identifier $replica$ and maps $sc$, $p$ and $q$.
    }
    \label{fig:epochchangeobj}
\end{figure}

\begin{figure}[H]
    \begin{minipage}[t]{0.48\textwidth}  
    \centering
        \begin{tikzpicture}[node distance=-\pgflinewidth]

\node[short,fill=black] (a) {};
\node[long,right=of a,label=center:leader] (b) {};
\node[long,right=of b, label=center:pack] (c) {};

\end{tikzpicture}
    \captionof{figure}{%
        Message format of \PREPREPARE. It simply contains the leader's unique identifier $leader$ and a \PACK (in its format described above).
    }
    \label{fig:preprepare_msg}
    \end{minipage}
    \hfill
    \begin{minipage}[t]{0.48\textwidth}  
    \centering
        \begin{tikzpicture}[node distance=-\pgflinewidth]

\node[short,fill=black] (a) {};
\node[long,right=of a, label=center:backup] (b) {};
\node[long,right=of b,label=center:hash] (c) {};
\lnode{c}{d}{p\_cert};

\end{tikzpicture}
    \captionof{figure}{%
        Message format of \PREPARE. It contains the replica's unique identifier $replica$, the hash of a \PACK and the replica's prepare certificate $p\_cert$.
    }
    \label{fig:prepare}
    \end{minipage}
\end{figure}

\begin{figure}[H]
    \begin{minipage}[t]{0.48\textwidth}  
    \centering
        \begin{tikzpicture}[node distance=-\pgflinewidth]

\node[short,fill=black] (a) {};
\node[long,right=of a, label=center:leader] (b) {};
\node[long,right=of b,label=center:hash] (c) {};
\lnode{c}{d}{p\_cert*};

\end{tikzpicture}
    \caption{%
        Message format of \PREPARED. It contains the leader's unique identifier $leader$, the hash of a \PACK and the threshold signature $p\_cert*$.
    }
    \label{fig:prepared}
    \end{minipage}
    \hfill
    \begin{minipage}[t]{0.48\textwidth}  
    \centering
        \begin{tikzpicture}[node distance=-\pgflinewidth]

\node[short,fill=black] (a) {};
\node[long,right=of a, label=center:backup] (b) {};
\node[long,right=of b,label=center:hash] (c) {};
\lnode{c}{d}{c\_cert};

\end{tikzpicture}
    \caption{%
        Message format of \COMMIT. It contains the replica's unique identifier $replica$, the hash of a \PACK and the replica's commit certificate $c\_cert$.
    }
    \label{fig:commit}
    \end{minipage}
\end{figure}

\begin{figure}[H]
    \begin{minipage}[t]{0.48\textwidth}  
    \centering
        \begin{tikzpicture}[node distance=-\pgflinewidth]

\node[short,fill=black] (a) {};
\node[long,right=of a, label=center:leader] (b) {};
\node[long,right=of b,label=center:hash] (c) {};
\lnode{c}{d}{c\_cert*};

\end{tikzpicture}
    \caption{%
        Message format of \COMMITTED. It contains the leader's unique identifier $leader$, the hash of a \PACK and the threshold signature $c\_cert*$.
    }
    \label{fig:committed}
    \end{minipage}
    \hfill
    \begin{minipage}[t]{0.48\textwidth}  
    \centering
        \begin{tikzpicture}[node distance=-\pgflinewidth]

\node[short,fill=black] (a) {};
\node[long,right=of a, label=center:client] (b) {};
\node[long,right=of b,label=center:request] (c) {};
\node[long, right=of c,label=center:r\_hash] (d) {};

\end{tikzpicture}
    \caption{%
        Message format of \REQUEST. It consists of the client's unique identifier $client$, the request $request$ and the hash of the request $r\_hash$.
    }
    \label{fig:request}
    \end{minipage}
\end{figure}

\begin{figure}[H]
    \begin{minipage}[t]{0.48\textwidth}  
    \centering
 \begin{tikzpicture}[node distance=-\pgflinewidth]

\node[short,fill=black] (a) {};
\node[long,right=of a, label=center:replica] (b) {};
\node[short, right=of b] (c) {\rotatebox{270}{sn}};
\node[short, right=of c] (d) {\rotatebox{270}{e}};
\node[long, right=of d,label=center:response] (e) {};
\node[long,right=of e,label=center:r\_hash] (f) {};

\end{tikzpicture}
    \caption{%
        Message format of \ANSWER. It consists of the replica's unique identifier $replica$, the sequence number $sn$ and epoch number $e$ assigned to the request, the response $response$ and the request's hash $r\_hash$.
    }
    \label{fig:answer}
    \end{minipage}
    \hfill
    \begin{minipage}[t]{0.48\textwidth}  
    \centering
        \begin{tikzpicture}[node distance=-\pgflinewidth]

\node[short,fill=black] (a) {};
\node[long,right=of a, label=center:backup] (b) {};
\node[long, right=of b, label=center:c\_sn] (c) {};
\node[long, right=of c, label=center:digest] (d) {};
\lnode{d}{e}{c\_cert};

\end{tikzpicture}
    \caption{%
        Message format of \NEWCHECKPOINT. It consists of the backup's unique identifier $backup$, the sequence number of the checkpoint $c\_sn$, the digest of all hashes of all \PACK objects in that checkpoint and the backup's signature $c\_cert$.
    }
    \label{fig:new_checkpoint}
    \end{minipage}
\end{figure}

\begin{figure}[H]
    \begin{minipage}[t]{0.48\textwidth}  
    \centering
        \begin{tikzpicture}[node distance=-\pgflinewidth]

\node[short,fill=black] (a) {};
\node[long,right=of a, label=center:leader] (b) {};
\node[long, right=of b, label=center:c\_sn] (c) {};
\node[long, right=of c, label=center:digest] (d) {};
\lnode{d}{e}{c\_cert*};

\end{tikzpicture}
    \caption{%
        Message format of \CHECKPOINTED. It consists of the leader's unique identifier $leader$, the sequence number of the checkpoint $c\_sn$, the digest of all hashes of all \PACK objects in that checkpoint and the threshold signature $c\_cert*$.
    }
    \label{fig:checkpointed}
    \end{minipage}
    \hfill
    \begin{minipage}[t]{0.48\textwidth}  
    \centering
        \begin{tikzpicture}[node distance=-\pgflinewidth]

\node[short,fill=black] (a) {};
\node[long,right=of a, label=center:eco] (b) {};
\lnode{b}{c}{ec\_cert};
\node[long,right=of c,label=center:e\_hash] (d) {};

\end{tikzpicture}
    \caption{%
        Message format of \EPOCHCHANGE. It consists of an epoch-change object $eco$ as described above, the primary's signature $ec\_cert$ and the hash of the \ECO $e\_hash$.
    }
    \label{fig:epochchange}
    \end{minipage}
\end{figure}

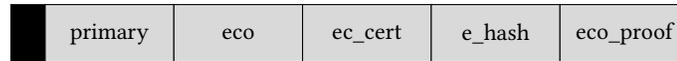
\begin{figure}[H]
    \centering
    \begin{tikzpicture}[node distance=-\pgflinewidth]

\node[short,fill=black] (a) {};
\node[long,right=of a, label=center:primary] (b) {};
\node[long,right=of b, label=center:eco] (c) {};
\lnode{c}{d}{ec\_cert};
\node[long,right=of d,label=center:e\_hash] (e) {};
\node[long,right=of e,label=center:eco\_proof] (f) {};

\end{tikzpicture}
    \caption{%
        Message format of \NEWEPOCH. It consists of the primary's unique identifier $primary$, an epoch-change object $eco$ as described above, the primary's certificate $ec\_cert$, the hash of the \ECO $e\_hash$ and a list of \ECO $eco\_proof$ that were used to compute $eco$.
    }
    \label{fig:newepoch}
\end{figure}
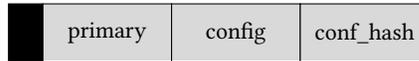
\begin{figure}[H]
    \centering
            \begin{tikzpicture}[node distance=-\pgflinewidth]

\node[short,fill=black] (a) {};
\node[long,right=of a, label=center:primary] (b) {};
\node[long,right=of b, label=center:config] (c) {};
\node[long,right=of c,label=center:conf\_hash] (d) {};

\end{tikzpicture}
    \caption{%
        Message format of \NEWEPOCHCONF. It consists of the primary's unique identifier $primary$, the new replica configuration $config$ and the configuration's hash $conf\_hash$.
    }
    \label{fig:newepochconfig}
\end{figure}

\subsection{Epoch Operation}
The following detail the protocol's behavior during epochs.

\paragraph{Client Request Pipeline.}
The normal operation of \prot is as follows:
\begin{enumerate}
    \item When a client broadcasts a request, each replica assigns the request to a bucket by hashing the issuer's \textit{unique client id}.
    \item The leader responsible for that bucket assigns a monotonically increasing \textit{sequence number} to the request and broadcasts a signed \PREPREPARE message.
    \item Backups receiving a \PREPREPARE message verify the message and reply with a signed \PREPARE message to the leader, using a partial threshold signature.
    \item Once a leader receives $2f$ valid \PREPARE messages for a given sequence number, it computes a combined signature from the $2f$ \PREPARE messages and its own signature. It then broadcasts the combined signature in a \PREPARED message to all replicas.
    \item A backup receiving a \PREPARED message, verifies the signature and replies with a \COMMIT message. 
    \item Upon receiving $2f$ \COMMIT messages for a sequence number, the leader once again computes a combined signature from its own signature and $2f$ \COMMIT messages and broadcasts the combined signature in a \COMMITTED message.
    \item When a backup receives a \COMMITTED message $c$ from leader $l$ and all preceding requests from $l$ have been delivered, it executes the client request associated with $c$'s sequence number.
    \item The client accepts a reply for a given request, once it received $f+1$ replies with the same result.
\end{enumerate}

\paragraph{Checkpoint Operation.}

Sequence numbers are assigned to leaders within given intervals (\textit{watermark bounds}). Additionally, leaders may only distribute a subset of these sequence numbers at any given point in time. This interval is a sliding window of sequence numbers (\textit{watermarks}) and updated using \textit{checkpoints}. The size of the interval is $2K$, with $K$ being a configuration parameter. Checkpoints prove the correctness of a given state by containing signatures of $2f+1$ replicas. The generation of checkpoints work as follows:

\begin{itemize}
    \item When a backup has $K$ continuous locally committed requests from a given leader $l$ since its last checkpoint, it will issue a \NEWCHECKPOINT message to the leader. In other words, when all requests with sequence numbers $sn$, such that $prev\_checkpoint\_sn <= sn < prev\_checkpoint\_sn + K$, where $prev\_checkpoint\_sn$ is the sequence number of the previous checkpoint of leader $l$, have been locally committed, the backup requests the generation of a new checkpoint for that leader.
    \item Upon receiving $2f$ \NEWCHECKPOINT messages for a given sequence number, the leader $l$ computes a combined signature out of these messages and its own signature and broadcasts it in a \CHECKPOINTED message.
    \item Replicas receiving a \CHECKPOINTED message verify its contents and advance the watermarks of leader $l$.
\end{itemize}

\subsection{Epoch-Change}

By rotating bucket assignments, \prot ensures progress and prevents request censoring by byzantine replicas. This bucket rotation takes place during a phase called \textit{epoch-change}. During this phase, no new client requests will effectively be treated.

\paragraph{Epoch-Change Initiation.}
Replicas will call for an epoch-change if at least one of the following conditions is met:

\begin{itemize}
    \item $2f+1$ leaders have exhausted all sequence numbers within their watermark bounds: a majority of correct leaders will not be able to make progress anymore; thus, replicas request an epoch-change.
    \item The \textit{epoch timer} runs out. This timer is set to a constant value and started at the beginning of each epoch.
    \item $2f+1$ \textit{no-progress timers} run out. Each replica has a set of $n$ no-progress timers. These timers are reset every time the replica commits a request from a given leader.
    \item A replica sees $f+1$ \EPOCHCHANGE messages for epochs higher than its current epoch.
\end{itemize}

In summary, replicas initiate epoch-changes either when the system no longer makes enough progress, or when the epoch's time runs out.

\paragraph{Epoch-Change Pipeline.}

The sequence of events in an epoch-change is as follows:
\begin{enumerate}
    \item When one of the above condition holds, replicas request an epoch-change to epoch $e+1$ by broadcasting an \EPOCHCHANGE message. They then set an \textit{epoch-change timer}. If the timer runs out before receiving a \NEWEPOCH message, the replica will request an epoch-change to epoch $e+2$ (and so on). Replicas increase their local epoch number and thus stop treating all messages related to old epochs.
    \item Once the primary for the current epoch (assigned in a round-robin manner) has received $2f$ \NEWEPOCH messages for epoch $e^*$, it will calculate the new replica configuration for $e^*$ and initiate a classical reliable 3-phase broadcast of the new replica configuration using \NEWEPOCHCONF messages. The new configuration takes into account past leader performances -- if a leader has successfully committed all assigned requests in the previous epoch, the watermark bounds are doubled, otherwise the new watermark bounds are equal to the number of previously processed requests. It then computes and broadcasts a \NEWEPOCH message for epoch $e^*$.
    \item Replicas receiving \NEWEPOCHCONF messages participate in the reliable broadcast and eventually adopt the new replica configuration for epoch $e^*$.
    \item Replicas receiving \NEWEPOCH messages verify the message's content by performing the same computation as the primary. Finally, replicas process the \NEWEPOCH message and resume normal operation by treating requests in their new bucket, according to the new replica configuration.
\end{enumerate}

\subsection{Format, Calculation and Processing of Epoch-Change Messages}
\label{sec:proc_newepoch}
The basis of \EPOCHCHANGE and \NEWEPOCH messages are \textit{epoch-change objects} (see Figure~\ref{fig:epochchangeobj}). 

\noindent Epoch-change objects in an \EPOCHCHANGE$<e, replica, sc, p, q>$ message contain:
\begin{itemize}
\item[\textbf{e}] The epoch-number that this epoch-change object relates to.
\item[\textbf{replica}] The unique replica ID of the owner/issuer of this object.
\item[\textbf{sc}] Map with keys \texttt{replica\_id} and values \texttt{<checkpoint\_sn,checkpoint\_cert>}, where \texttt{checkpoint\_sn} is the sequence number of the most recent checkpoint seen for \texttt{replica\_id} and \texttt{checkpoint\_cert} the corresponding combined signature verifying the checkpoint's validity.
\item[\textbf{p}] Map with keys \texttt{replica\_id} and values \texttt{map<sn, p\_certs>}, where \texttt{sn} is the sequence number of a prepared request (received \PREPARED) that has not been committed and for which \texttt{replica\_id} is the leader and \texttt{p\_certs} the combined signature of the corresponding \PREPARED message.
\item[\textbf{q}] Map with keys \texttt{replica\_id} and values \texttt{map<sn, q\_certs>}, where \texttt{sn} is the sequence number of a committed request (received \COMMITTED) that is not part of any checkpoints and for which \texttt{replica\_id} is the leader and \texttt{q\_certs} the combined signature of the corresponding \COMMITTED message.
\end{itemize}

Epoch-change objects \texttt{eco*} in a \NEWEPOCH message contain the same elements; however, their content are the "supersets" (denoted by \texttt{sc*},\texttt{p*}, \texttt{q*}) of the contents of all $2f+1$ \EPOCHCHANGE messages necessary for the generation of \NEWEPOCH messages. Specifically, the primary performs the following calculations.

To calculate \texttt{sc*}, for each replica $r$, the primary chooses the highest sequence number $sn'_r$ and corresponding certificate among all sets \texttt{sc} received through \EPOCHCHANGE messages and adds them to \texttt{sc*}. For all requests with sequence numbers $sn$ higher than $sn'_r$ belonging to $r$:
\begin{itemize}
    \item If there is a set \texttt{q'} among all sets \texttt{q} received through \EPOCHCHANGE messages that contains $sn$, $sn$ and the corresponding \texttt{q\_cert} are added to \texttt{q*}.
    \item If there is a set \texttt{p'} among all sets \texttt{p} received through \EPOCHCHANGE messages that contains $sn$ and no set \texttt{q'} among all sets \texttt{q} received through \EPOCHCHANGE messages contains $sn$, $sn$ and the corresponding \texttt{p\_cert} are added to \texttt{p*}.
\end{itemize}
Note, that both conditions are mutually exclusive. 

The \NEWEPOCH message contains an epoch-change object \texttt{eco*}, a set \texttt{eco\_proof} of $2f+1$ \EPOCHCHANGE messages ($2f$ received from replicas and the primary's) from which the \NEWEPOCH message was calculated, the message's hash \texttt{e\_hash} and the primary's certificate \texttt{ec\_cert}.

Replicas process \NEWEPOCH messages as follows:
\begin{itemize}
    \item For every element in \texttt{q*}, replicas re-send answers of requests back to clients.
    \item For every element in \texttt{p*}, replicas generate a new \PREPREPARE message with a reference (see \texttt{opt\_ref} in Figure~\ref{fig:pack}) to the original request from the previous epoch. The new message is a regular \PREPREPARE message and will go through the pipeline like any other message, but when replicas execute the request after receiving a \COMMITTED message, they will execute requests from the reference instead.
    \item For every sequence number $sn$ that is within watermarks of each leader, replicas generate \textit{null requests} and a new \PREPREPARE message with those requests as \texttt{opt\_ref}. Null requests are requests, whose execution results in no-ops. These messages will go through the pipeline like any other message.
\end{itemize}

\end{document}